\documentclass[a4paper,UKenglish,numberwithinsect]{article}

\usepackage{microtype}%

\usepackage{color}
\usepackage{gensymb}
\usepackage[inline]{enumitem}
\usepackage{todonotes}
\usepackage{hyperref}
\usepackage{sidecap}
\usepackage{wrapfig}
\usepackage{cleveref}
\usepackage{geometry}

\usepackage{amsmath}
\usepackage{amsthm}
\usepackage{amssymb}
\usepackage{amsopn}
\usepackage{subfig}
\usepackage[utf8]{inputenc}

\theoremstyle{plain}
\newtheorem{theorem}{Theorem}
\newtheorem{lemma}[theorem]{Lemma}
\newtheorem{corollary}[theorem]{Corollary}
\theoremstyle{definition}
\newtheorem{definition}[theorem]{Definition}

\title{Towards a Topology-Shape-Metrics Framework for Ortho-Radial Drawings}

\author{Lukas Barth \and Benjamin Niedermann \and Ignaz Rutter \and Matthias 
Wolf}

\newtheorem{observation}[theorem]{Observation}

\newcommand{\widebar}[1]{\mkern 
1.5mu\overline{\mkern-1.5mu#1\mkern-1.5mu}\mkern 1.5mu}
\newcommand{\subpath}[2]{\ensuremath{#1[#2]}}

\newcommand{\N}{\ensuremath{\mathbb{N}}}
\newcommand{\join}{+}
\newcommand{\reverse}[1]{\ensuremath{\widebar{#1}}}
\newcommand{\NP}{\ensuremath{\mathcal{NP}}}

\DeclareMathOperator{\rot}{rot}

\begin{document}

\maketitle

\begin{abstract}
  \emph{Ortho-Radial drawings} are a generalization of orthogonal
  drawings to grids that are formed by concentric circles and
  straight-line spokes emanating from the circles' center. Such
  drawings have applications in schematic graph layouts, e.g., for
  metro maps and destination maps.

  A plane graph is a planar graph with a fixed planar embedding.  We
  give a combinatorial characterization of the plane graphs that admit
  a planar ortho-radial drawing without bends.  Previously, such a
  characterization was only known for paths, cycles, and theta
  graphs~\cite{hht-orthoradial-09}, and in the special case of
  rectangular drawings for cubic graphs~\cite{hhmt-rrdcp-10}, where
  the contour of each face is required to be a rectangle.

  The characterization is expressed in terms of an \emph{ortho-radial
    representation} that, similar to Tamassia's \emph{orthogonal
    representations} for orthogonal drawings describes such a drawing
  combinatorially in terms of angles around vertices and bends on the
  edges.  In this sense our characterization can be seen as a first
  step towards generalizing the Topology-Shape-Metrics framework of
  Tamassia to ortho-radial drawings.
\end{abstract}

\begin{figure}[t]
  \centering
  \subfloat[Ortho-Radial
  Grid.]{ 
  \label{fig:pre:drawing-grid}\includegraphics[scale=1.3]{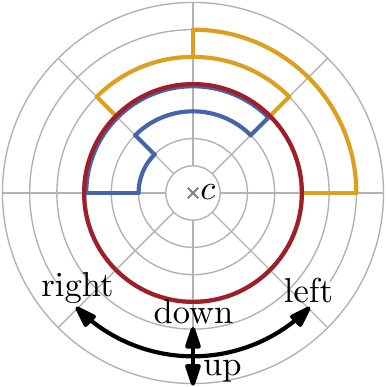}
  }%
  \hspace{10ex}
  \subfloat[Cylinder.]{ 
  \label{fig:pre:drawing-cylinder}\includegraphics[scale=1.3]{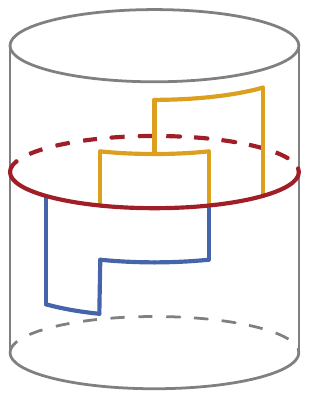}
  }
  \hfill
  \caption{An ortho-radial drawing of a graph on a grid
    \protect\subref{fig:pre:drawing-grid} and its equivalent interpretation
    as an orthogonal drawing on a cylinder
    \protect\subref{fig:pre:drawing-cylinder}.}
  \label{fig:pre:drawing}
\end{figure}

\begin{figure}[t]
  \centering
  \includegraphics[width=0.7\textwidth]{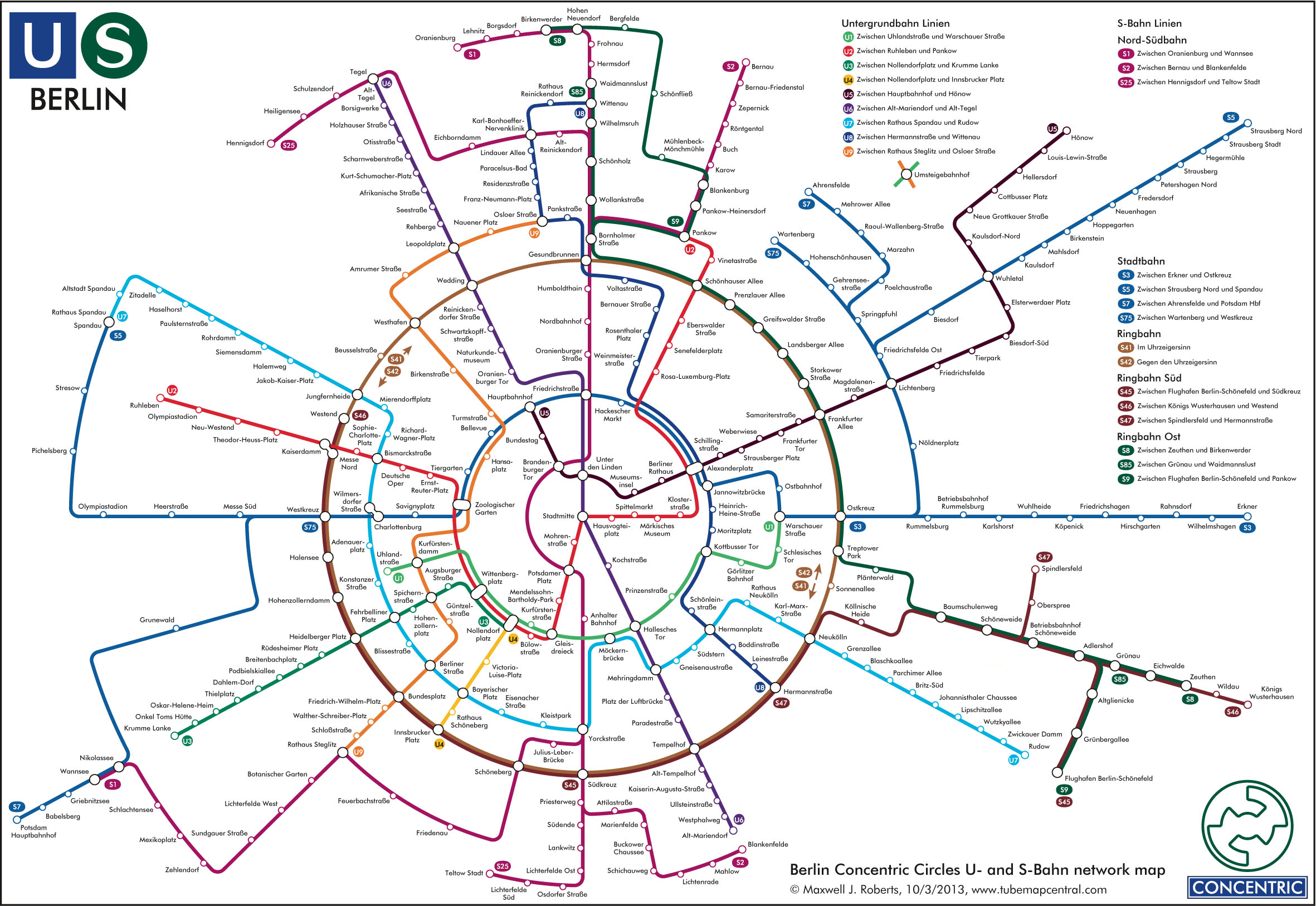}

  \caption{Metro map of Berlin using an ortho-radial layout. Image
    copyright by Maxwell J.~Roberts.}
    \label{fig:intro:berlin}
\end{figure}

\section{Introduction}
\label{sec:introduction}
\emph{Grid drawings} of graphs map vertices to grid points, and edges to internally disjoint curves on the grid lines connecting their endpoints. The appropriate choice of the underlying
grid is decisive for the quality and properties of the drawing.
\emph{Orthogonal grids}, where the grid lines are horizontal and
vertical lines, are popular and widely used in graph drawing.  Their
strength lies in their simple structure, their high angular
resolution, and the limited number of directions.  Graphs admitting
orthogonal grid drawings must be \emph{4-planar}, i.e., they must be
planar and have maximum degree~4.  On such grids a single edge
consists of a sequence of horizontal and vertical grid segments.  A
transition between a horizontal and a vertical segment on an edge is
called a \emph{bend}. A typical optimization goal is the minimization
of the number of bends in the drawing, either in total or by the
maximum number of bends per edge.

The popularity and usefulness of orthogonal drawing have their
foundations in the seminal work of Tamassia~\cite{t-emn-87} who
showed that for \emph{plane graphs}, i.e., planar graphs with a fixed
embedding, a bend-minimal planar orthogonal drawing can be computed efficiently.  More generally,
Tamassia~\cite{t-emn-87} established the so-called
Topology-Shape-Metrics framework, abbreviated as TSM in the following,
for computing orthogonal drawings of 4-planar graphs.

The goal of this work is to provide a similar framework for
\emph{ortho-radial drawings}, which are based on ortho-radial grids rather
than orthogonal grids. Such drawings have applications in schematic
graph layouts, e.g., for metro maps and destination maps; see
Fig.~\ref{fig:intro:berlin} for an example.  An \emph{ortho-radial
  grid} is formed by $M$ concentric circles and $N$ spokes, where
$M,N\in\N$.  More precisely, the radii of the concentric circles are
integers, and the $N$ spokes emanate from the center~$c=(0,0)$ of the
circles such that they have uniform angular resolution; see
Fig.~\ref{fig:pre:drawing}. We note that $c$ does not belong to the
grid.  Again vertices are positioned at grid points and edges are
drawn as internally disjoint chains of (1) segments of spokes and (2)
arcs of concentric circles.  As before, a \emph{bend} is a transition between a spoke segment
and an arc segment.  We observe that ortho-radial
drawings can also be thought of as orthogonal drawings on a cylinder;
see Fig.~\ref{fig:pre:drawing}.  Edge segments that are originally
drawn on spokes are parallel to the axis of the cylinder, whereas
segments that are originally drawn as arcs are drawn on circles
orthogonal to the axis of the cylinder.

It is not hard to see that every orthogonal drawing can be transformed
into an ortho-radial drawing with the same number of bends.  The
converse is, however, not true.  It is readily seen that for example a
triangle, which requires at least one bend in an orthogonal drawing,
admits an ortho-radial drawing without bends, namely in the form of a
circle centered at the origin.  In fact, by suitably nesting
triangles, one can construct graphs that have an ortho-radial drawing
without bends but require a linear number of bends in any orthogonal
drawing. Similarly, the octahedron, which requires an edge with three
bends in a planar orthogonal drawing has an ortho-radial drawing with
two bends per edge; see Fig.~\ref{fig:intro:octahedron}.

\begin{figure}[tb]
  \centering
  \includegraphics[scale=1.2]{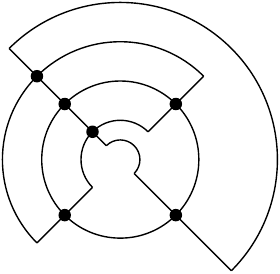}
  \hspace{6ex}
  \includegraphics{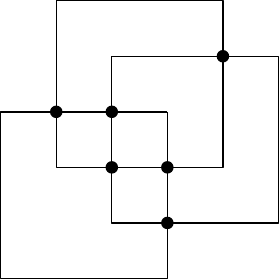}
  \captionof{figure}{Bend-minimal drawings of an octahedron graph. In the ortho-radial drawing, every edge has at most two bends, while in the orthogonal drawing, one edge requires three bends. Also, the ortho-radial drawing requires just eight bends, while the orthogonal drawing requires twelve bends in total.}
 \label{fig:intro:octahedron}
\end{figure}

\subparagraph*{Related Work.}

The case of planar orthogonal drawings is well-researched and there is
a plethora of results known; see~\cite{dg-popda-13} for a survey.
Here, we mention only the most recent results and those which are most
strongly related to the problem we consider.  A $k$-embedding is an
orthogonal planar drawing such that each edge has at most $k$ bends.
It is known that, with the single exception of the octahedron graph, every planar graph has a 2-embedding and that
deciding the existence of a 0-embedding is \NP-complete \cite{gt-ccurpt-01}.  The latter in
particular implies that the problem of computing a planar orthogonal
drawing with the minimum number of bends is \NP-hard.  In contrast, the
existence of a 1-embedding can be tested
efficiently~\cite{bkrw-odfc-14}; this has subsequently been
generalized to a version that forces up to $k$ edges to have no bends
in FPT time w.r.t $k$~\cite{blr-ogdie-14}, and to an optimization
version that optimizes the number of bends beyond the first one on
each edge~\cite{brw-oodc-13}.  Moreover, for series-parallel graphs
and graphs with maximum degree~3 an orthogonal planar drawing with the
minimum number of bends can be computed
efficiently~\cite{blv-sood-98}.

The \NP-hardness of the bend minimization problem has also inspired the
study of bend minimization for planar graphs with a fixed embedding.
In his fundamental work Tamassia~\cite{t-emn-87} showed that for plane
graphs, i.e., planar graphs with a fixed planar embedding,
bend-minimal planar orthogonal drawing can be computed in polynomial
time.  The running time has subsequently been improved to
$O(n^{1.5})$~\cite{ck-am-12}.  Key to this result is the existence of
a combinatorial description of planar orthogonal drawings of a plane
graph in terms of the angle surrounding each vertex and the order and
directions of bends on the edges but neglecting any kind of geometric
information such as coordinates and edge lengths.  In particular, this
allows to efficiently compute bend-minimal orthogonal drawings of
plane graphs by mapping the purely combinatorial problem of
determining an orthogonal representation with the minimum number of
bends to a flow problem.

\begin{wrapfigure}[15]{r}{0.33\textwidth}
  \centering
  \includegraphics[scale=1.2]{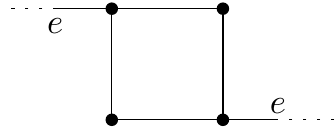}
  \caption{In this drawing, the angles around vertices sum up to $360^\circ$, and also the sum of angles for each face is as expected for an ortho-radial drawing. However, the graph does not have an ortho-radial drawing without bends.}
  \label{fig:intro:noteasy}
\end{wrapfigure}

In addition, there is a number of works that seek to characterize the
plane graphs that can be drawn without bends.  In this case an
orthogonal representation essentially only describes the angles around
the vertices.  Rahman, Nishizeki and Naznin~\cite{rnn-odpgb-03}
characterize the plane graphs with maximum degree~3 that admit such a
drawing and Rahman, Nakano and Nishizeki~\cite{rnn-rgdpg-98} characterize the plane graphs
that admit a rectangular drawing where in addition the contour of each
face is a rectangle.

In the case of ortho-radial drawings much less is known.  A natural
generalization of the properties of an orthogonal representation to the
ortho-radial case, where the angles around the vertices and inside the
faces are constrained, has turned out not be sufficient for drawability; see
Fig.~\ref{fig:intro:noteasy}. So far characterizations of bend-free
ortho-radial drawing have
only been achieved for paths, cycles, and theta
graphs~\cite{hht-orthoradial-09}.  For the special case of rectangular
ortho-radial drawings, i.e., every internal face is bounded by a rectangle, a
characterization is known for cubic graphs~\cite{hhmt-rrdcp-10}.

\subparagraph*{Contribution and Outline.}

Since deciding whether a 4-planar graph can be orthogonally drawn in
the plane without any bends is \NP-complete~\cite{gt-ccurpt-01}, it is
not surprising that also ortho-radial bend minimization is \NP-hard; see 
Section~\ref{sec:bend-minimization}.

As our main result we introduce a generalization of an orthogonal
representation, which we call \emph{ortho-radial representation}, that
characterizes the 4-plane graphs having bend-free ortho-radial drawings; see Theorem~\ref{thm:repr:characterization}.

This significantly generalizes the corresponding results for paths, cycles,
theta graphs~\cite{hht-orthoradial-09}, and cubic graphs~\cite{hhmt-rrdcp-10}.
Further, this characterization can be seen as a step towards an extension of the
TSM framework for computing ortho-radial drawings that may have bends. Namely,
once the angles around vertices and the order and directions of bends along each
edge of a graph $G$ have been fixed, we can replace each bend by a vertex to
obtain a graph $G'$.  The directions of bends and the angles at the vertices of
$G$ define a unique ortho-radial representation of $G'$, which is valid if and
only if $G$ admits a drawing with the specified angles and bends. Thus,
ortho-radial drawings can indeed be described by angles around vertices and
orders and directions of bends on edges. Our main result therefore implies that
ortho-radial drawings can be computed by a TSM framework, i.e., by fixing a
combinatorial embedding in a first ``Topology'' step, determining a description
of the drawing in terms of angles and bends in a second ``Shape'' step, and
computing edge lengths and vertex coordinates in a final ``Metrics'' step.

In the following, we disregard the ``Topology'' step and assume that
our input consists of a 4-planar graph with a fixed
\emph{combinatorial embedding}, i.e., the order of the incident edges
around each vertex is fixed and additionally, one outer face and one
central face are specified; the latter shall contain the center of the
drawing.  We present our definition of an ortho-radial representation
in Section~\ref{sec:ortho-radial-representations}.  After introducing
some basic tools based on this representation in
Section~\ref{sec:preliminaries}, we first present in
Section~\ref{sec:characterization-rect} a characterization for
rectangular graphs, whose ortho-radial representation is such that
internal faces have exactly four $90^\circ$ angles, while all other
incident angles are $180^\circ$; i.e., they have to be drawn as
rectangles.

The algorithm we use as a proof of drawability corresponds to the
``Metrics'' step of an ortho-radial TSM framework.  The
characterization corresponds to the output of the ``Shape''
step. Based on that special case of 4-planar graphs, we then present
the characterization and the ``Metrics'' step for general 4-planar
graphs in Section~\ref{sec:rectangulation}.

\section{Ortho-Radial Drawings}

In this section we introduce basic definitions and conventions on
ortho-radial drawings that we use throughout this paper. To that end,
we first introduce some basic definitions.

We assume that paths cannot contain vertices multiple times but cycles
can, i.e. all paths are simple. We always consider cycles that are
directed clockwise, so that their interior is locally to the right of
the cycle. A cycle is part of its interior and its exterior.

For a path $P$ and vertices $u$ and $v$ on $P$, we denote the subpath
of $P$ from $u$ to $v$ (including these vertices) by
$\subpath{P}{u,v}$. We may also specify the first and last edge instead and
write, e.g., $\subpath{P}{e,e'}$ for edges $e$ and $e'$ on $P$. We denote the
concatenation of two paths $P$ and $Q$ by $P + Q$.
For any path $P=v_1\dots v_k$, we
define its reverse $\reverse{P}=v_k\dots v_1$. For a cycle~$C$
that contains any edge at most once (e.g.~if $C$ is simple), we extend
the notion of subpaths as follows: For two edges $e$ and $e'$ on $C$,
the subpath $\subpath{C}{e,e'}$ is the unique path on $C$
that starts with $e$ and ends with $e'$.  If the start vertex $v$ of
$e$ identifies $e$ uniquely, i.e., $C$ contains $v$ exactly once, we
may write $\subpath{C}{v, e'}$ to describe the path on $C$ from $v$ to
$e'$. Analogously, we may identify $e'$ with its endpoint if this is
unambiguous.

We are now ready to introduce concepts of ortho-radial drawings.
We are given a $4$-planar graph $G=(V,E)$ with fixed embedding. We refer
to the directed edge from $u$ to $v$ by $uv$.  Consider an
ortho-radial drawing~$\Delta$ of $G$. Recall that we allow no bends on edges.
In $\Delta$ a directed edge~$e$ is either drawn clockwise, counter-clockwise,
towards the center or away from the center; see Fig.~\ref{fig:pre:drawing}.
Using the
metaphor of drawing $G$ on a cylinder, we say that $e$ points
\emph{right}, \emph{left}, \emph{down} or \emph{up},
respectively. Edges pointing left or right are \emph{horizontal} edges
and edges pointing up or down are \emph{vertical} edges.

There are two fundamentally different ways of drawing a simple cycle~$C$: The
center of the grid may lie in the interior or the
exterior of $C$. In the former case $C$ is \emph{essential} and in the
latter case it is \emph{non-essential}. In Fig.~\ref{fig:pre:drawing}
the red cycle is essential and the blue cycle is non-essential.

We further observe that~$\Delta$ has two special faces: One unbounded
face, called the \emph{outer face}, and the \emph{central face}
containing the center of the drawing.  These two faces are equal if
and only if $\Delta$ contains no essential cycles. All other faces of
$G$ are called \emph{regular}.  Ortho-radial drawings without
essential cycles are equivalent to orthogonal
drawings~\cite{hht-orthoradial-09}.  That is, any such ortho-radial
drawing can be transformed to an orthogonal drawing of the same graph
with the same outer face and vice versa.

We represent a face as a cycle $f$ in which the interior of the face
lies locally to the right of~$f$. Note that $f$ may not be simple
since cut vertices may appear multiple times on $f$. But no directed
edge is used twice by $f$.  Therefore, the notation of subpaths of
cycles applies to faces.  Note furthermore that the cycle bounding the
outer face of a graph is directed counter-clockwise, whereas all other
faces are bounded by cycles directed clockwise.

A face~$f$ in $\Delta$ is a \emph{rectangle} if and only if its
boundary does not make any left turns.  That is, if $f$ is a regular
face, there are exactly $4$ right turns, and if $f$ is the central or
the outer face, there are no turns at all.  Note that by this
definition $f$ cannot be a rectangle if it is both the outer and the
central face.

\section{Ortho-Radial Representations}\label{sec:ortho-radial-representations}

In this section, we define ortho-radial representations. These are a tool to describe
the ortho-radial drawing of a graph without fixing any edge lengths. As mentioned in the introduction, they are an analogon to orthogonal representations in the TSM framework.

We observe all directions of all edges being fixed once all angles around vertices are fixed.

For two edges $uv$ and $vw$ that enclose the
angle~$\alpha\in\{90\degree,180\degree,270\degree,360\degree\}$ at $v$
(such that the angle measured lies locally to the right of $uvw$), we
define the rotation~$\rot(uvw)=2-\alpha/90\degree$.  We note that the
rotation is 1 if there is a right turn at $v$, $0$ if $uvw$ is
straight, and $-1$ if a left turn occurs at $v$. If $u=w$, it is
$\rot(uvw)=-2$.

We define the rotation of a path $P=v_1\dots v_k$ as the sum of the rotations at its internal vertices, i.e.,
$
\rot(P)=\sum_{i=2}^{k-1}\rot(v_{i-1}v_iv_{i+1}).
$
Similarly, for a cycle $C=v_1\dots v_kv_1$, its rotation is the sum of the rotations at all its vertices (where we define $v_0=v_k$), i.e.,
$
\rot(C)=\sum_{i=1}^{k}\rot(v_{i-1}v_iv_{i+1}).
$

When splitting a path at an edge, the sum of the rotations of the two
parts is equal to the rotation of the whole path.
\begin{observation}\label{obs:rot_splitting_path}
  Let $P$ be a simple path with start vertex $s$ and end vertex
  $t$. For all edges $e$ on $P$ it holds that $ \rot(P) =
  \rot(\subpath{P}{s, e}) + \rot(\subpath{P}{e,t})$.
\end{observation}

Furthermore, reversing a path changes all left turns to right turns and vice versa. Hence, the sign of the rotation is flipped.
\begin{observation}\label{obs:rot_reverse}
For any path~$P$ it holds $
\rot(\reverse{P}) = -\rot(P)$.
\end{observation}

The next observation analyzes \emph{detours}; see also Figure~\ref{fig:pre:rot_path_detour}.

\begin{figure}[tb]
 \centering
 \includegraphics{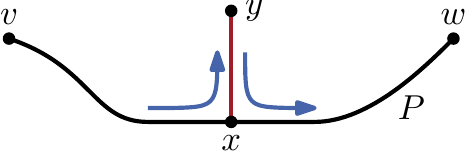}
 \hspace{1ex}
 \includegraphics{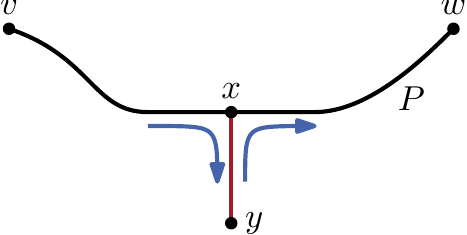}
 \captionof{figure}{The situation of Observation~\ref{obs:pre:rot_path_detour}.
 If one makes a detour via $xy$ (without accounting for the $360\degree$-turn at $y$), the rotation changes by $-2$, if $xy$ is to the left of $P$, or $+2$, if $xy$ lies to the right.}
 \label{fig:pre:rot_path_detour}
\end{figure}

\begin{observation}\label{obs:pre:rot_path_detour}
Let $P$ be a path from $v$ to $w$ and $xy$ an edge not on $P$ such that $x$ is
an internal vertex of $P$.
If $xy$ lies locally to the left of $P$,
$
\rot\left(\subpath{P}{v,x}\join xy\right) + \rot\left(yx+\subpath{P}{x,w}\right) = \rot\left(P\right) - 2.
$
Further, if $xy$ lies locally to the right of $P$,
$
\rot\left(\subpath{P}{v,x}\join xy\right) + \rot\left(yx+\subpath{P}{x,w}\right) = \rot\left(P\right) + 2.
$
\end{observation}

\begin{SCfigure}[50][tb]
    \centering
    \includegraphics[page=2]{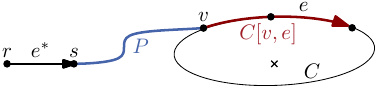}
       \caption{The labeling of $e$ induced by $P$ is
      $\ell^P_C(e)=\rot(e^\star\join P\join \subpath{C}{v,e})$.}
    \label{fig:repr:labeling}
    
  \end{SCfigure}

An \emph{ortho-radial representation} of a 4-planar graph~$G$ consists of a list
$H(f)$ of pairs $(e, a)$ for each face $f$, where $e$ is an edge on~$f$, and
$a\in\{90,180,270,360\}$.  Further, the outer face and the central face are
fixed and one \emph{reference edge} $e^\star$ in the outer face is given, with
$e^\star$ oriented such that the outer face is locally to its left. By
convention the
reference edge is always drawn such that it points right.
We interpret the fields of a pair in $H(f)$ as follows:
$e$ denotes the edge on $f$ directed such that $f$ lies to the right
of $e$. The field $a$ represents the
angle inside $f$ from $e$ to the following edge in degrees. Using this
information we define the \emph{rotation} of such a pair $t=(e,a)$
as $\rot(t)={(180-a)}/{90}$.

Not every ortho-radial representation also yields a valid ortho-radial
drawing of a graph. In order to characterize valid ortho-radial
representations we introduce \emph{labelings} of essential cycles.
These labelings prove to be a valuable tool to ensure that the
drawings of all cycles of the graph are compatible with each other.

Let $G$ be an embedded 4-planar graph and let~$e^\star=rs$ be a
reference edge of~$G$.  Further, let $C$ be a simple, essential cycle
in $G$, and let $P$ be a path from $s$ to a vertex $v$ on $C$.  The
\emph{labeling}~$\ell^P_C$ of $C$ induced by $P$ is defined for each
edge~$e$ on~$C$ by $\ell^P_C(e)=\rot(e^\star\join
P\join\subpath{C}{v,e})$; see Fig.~\ref{fig:repr:labeling} for an
illustration.  We are mostly interested in labelings that are induced
by paths starting at $s$ and intersecting $C$ only at their endpoints. We 
call such paths \emph{elementary paths}.

We now introduce properties characterizing bend-free ortho-radial drawings, as we prove in Theorem~\ref{thm:repr:characterization}.

\begin{definition}\label{def:repr:valid_representation}
An ortho-radial representation is \emph{valid}, if the following conditions hold:
\begin{enumerate}
\item\label{cond:repr:sum_of_angles} The angle sum of all edges around each vertex given by the $a$-fields is 360.
\item\label{cond:repr:rotation_faces} For each face $f$, it is
\[
\rot(f)=
\begin{cases}
  4, & \text{$f$ is a regular face} \\
  0, & \text{$f$ is the outer or the central face but not both} \\
  -4,& \text{$f$ is both the outer and the central face} \\
\end{cases}
\]
\item\label{cond:repr:labeling} For each simple, essential cycle $C$ in $G$,
there is a labeling $\ell^P_C$ of $C$ induced by an elementary path $P$ such
that either $\ell^P_C(e)=0$ for all edges~$e$ of $C$, or there are edges $e_+$
and $e_-$ on $C$ such that $\ell^P_C(e_+)>0$ and $\ell^P_C(e_-)<0$.
\end{enumerate}
\end{definition}

The first two conditions ensure that the angle-sums at vertices
and inside the faces are correct.
Since the labels of neighboring edges differ by at most $1$ the last condition
ensures that on each essential cycle with not only horizontal edges there are
edges with labels $1$ and $-1$, i.e., edges pointing up and down. This reflects
the characterization of cycles~\cite{hht-orthoradial-09}. Basing all labels
on the reference edge guarantees that all cycles in the graph can be drawn
together consistently.

For an essential cycle~$C$ not satisfying the last
condition there are two possibilities: Either all labels of edges on $C$ are
non-negative and at least one label is positive, or all are non-positive and
at least one is negative. In the former case $C$ is
called \emph{decreasing} and in the latter
case it is \emph{increasing}. We call both \emph{monotone} cycles. Cycles with only the label 0 are not monotone.

We show that a graph with a given ortho-radial representation can be
drawn if and only if the representation is valid, which yields our main result:

\begin{theorem}\label{thm:repr:characterization}
  A 4-plane graph admits a bend-free orthogonal drawing if and only if
  it admits a valid ortho-radial representation.
\end{theorem}

In the next section, we introduce further tools based on
labelings. Subsequently, in Section~\ref{sec:characterization-rect} we prove Theorem~\ref{thm:repr:characterization} for
rectangular graphs. In Section~\ref{sec:rectangulation}, we generalize the result to
4-planar graphs.

\section{Properties of Labelings}\label{sec:preliminaries}
In this section we study the properties of labelings in more detail to
derive useful tools for proving
Theorem~\ref{thm:repr:characterization}.  Throughout this section, $G$
is a 4-planar graph with ortho-radial representation~$\Gamma$ that
satisfies
Conditions~\ref{cond:repr:sum_of_angles} and \ref{cond:repr:rotation_faces}
of Definition~\ref{def:repr:valid_representation}. Further let $e^\star$ be a
reference edge.
From Condition~\ref{cond:repr:rotation_faces} we obtain that the rotation of
all essential cycles is 0:
\begin{observation}\label{obs:rot_essential_cycle}
For any simple essential cycle $C$ of $G$ it is $\rot(C) = 0$.
\end{observation}

An ortho-radial representation fixes the direction of all edges: For the direction of an edge $e$, we consider a path~$P$ from
the reference edge to $e$ including both edges.  Different such paths
may have different rotations but we observe that these rotations
differ by a multiple of 4.
An edge $e$ is directed right, down, left, or up, if $\rot(P)$ is
congruent to 0, 1, 2, or 3 modulo 4, respectively.  If $e$ lies on an
essential cycle~$C$, we consider the label~$\ell^Q_C(e)$ induced by
any path~$Q$ instead of a path $P$, as the labels
are defined as rotations of such paths. Note that by this definition the
reference edge always points to the right.

\noindent{}Because the rotation of essential cycles is $0$ by
Observation~\ref{obs:rot_essential_cycle}, for two edges on an essential cycle
we observe:

\begin{observation}\label{obs:repr:label_difference}
  For any path $P$ and for any two edges $e$ and $e'$ on a simple,
  essential cycle~$C$, it holds that  $\rot(\subpath{C}{e,e'})=\ell^P_C(e')-\ell^P_C(e)$
\end{observation}

We now show that two elementary paths to the same essential
cycle~$C$ induce identical labelings of $C$.  Two paths $P$ and $Q$
from $e^\star$ to vertices on $C$ are \emph{equivalent} ($P\equiv_C Q$) if the corresponding labelings agree on all edges
of $C$, i.e., $\ell^P_C(e)=\ell^Q_C(e)$ for every edge $e$
of $C$. %

First, we prove that two paths are equivalent if the labels for just one
edge are equal.

\begin{figure}[t]
  \centering
  \subfloat[The edge $xy$ does not
  lie on $\subpath{C}{p,w}$ and $P'$ contains no duplicate
  edges.]{ \includegraphics{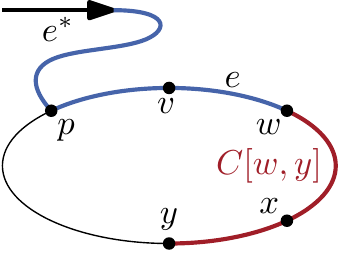}
    \label{fig:repr:equal_labels_edge-after}
  }
  \hspace{4em}
  \subfloat[The edge $xy$ lies on
 $\subpath{C}{p,w}$ and $P'$ completely goes around $C$. In both
 cases it is $\ell^P_C(xy)=\rot(e^\star\join P')$.]{ \includegraphics{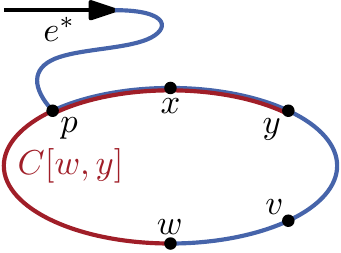}
    \label{fig:repr:equal_labels_edge-on}
  }

  \caption{Illustration of Lemma~\ref{lem:repr:equal_label_all_equal}.}
  \label{fig:repr:equal_labels_edge}
\end{figure}

\begin{lemma}\label{lem:repr:equal_label_all_equal}
Let $C$ be a simple, essential cycle and $P$ and $Q$ two paths from
the endpoint of $e^\star$ to vertices on~$C$. If there is an edge $e$
on $C$ with $\ell^P_C(e)=\ell^Q_C(e)$, then $P \equiv_{C} Q$.
\end{lemma}
\begin{proof}
  Denote the endpoints of $P$ and $Q$ on $C$ by $p$ and $q$,
respectively. Let $e=vw$ be an edge on $C$ such that
$\ell^P_C(vw)=\ell^Q_C(vw)$, which implies by the definition of the
labeling
\begin{equation}\label{eqn:repr:equal_labels_all_equal:first}
\rot(e^\star\join P\join\subpath{C}{p,vw})=\rot(e^\star\join
Q\join\subpath{C}{q,vw}).
\end{equation} For any edge $e'=xy$ on $C$, we obtain
$P'=P\join\subpath{C}{p,vw}\join\subpath{C}{w,y}$ and
$Q'=Q\join\subpath{C}{q,vw}\join\subpath{C}{w,y}$ by adding
$\subpath{C}{w,y}$ to the end of the two paths to $e$.
Equation~\ref{eqn:repr:equal_labels_all_equal:first} then implies
\begin{equation}\label{eqn:repr:equal_labels_all_equal:second}
\rot(e^\star\join P')=\rot(e^\star\join Q').
\end{equation}
If $w\in\subpath{C}{p,y}$ (as illustrated in
Fig.~\ref{fig:repr:equal_labels_edge-after}), it is
$\subpath{C}{p,y}=\subpath{C}{p,vw}\join\subpath{C}{w,y}$ and
therefore $\ell^P_C(e')=\rot(e^\star\join P')$.  Otherwise, $P'$
contains $C$ completely and $\rot(e^\star\join
P')=\rot(e^\star\join\subpath{C}{p,y})+\rot(C)=\ell^P_C(e')$, since
$C$ is essential (i.e., $\rot(C)=0$). This case is shown in
Fig.~\ref{fig:repr:equal_labels_edge-on}.
Analogously, we obtain $\ell^Q_C(e')=\rot(e^\star\join Q')$. Together
with Equation~\ref{eqn:repr:equal_labels_all_equal:second} we obtain
$\ell^P_C(e')=\rot(e^\star\join P')=\rot(e^\star\join
Q')=\ell^Q_C(e')$, which implies $P\equiv_C Q$.
\end{proof}

\begin{SCfigure}[50][bt]
 \centering
 \includegraphics{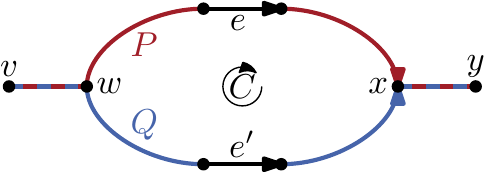}
 \caption{The paths $P$ and $Q$ use different sides of the cycle~$C$. The edges $e$ and $e'$ are chosen arbitrarily and cut $P$ and $Q$ in two parts.}
 \label{fig:repr:paths_cycle}
\end{SCfigure}

For a further useful property of labelings, we consider two paths
which traverse the same cycle $C$. Take a cycle $C$ and cut it open at
two vertices into two parts $C'$ and $C''$ and consider a path
traversing $C$ from cut-point to cut-point; see
Fig.~\ref{fig:repr:paths_cycle}. The path either completely includes
$C'$ or $C''$; in the former case, we call the path $P$, in the latter
case $Q$. We now show that both paths have the same rotation under
certain conditions.

\newcommand{\lemmaRotPathsCycle}{
Let $C$ be a simple cycle in $G$. Consider two vertices $w$ and $x$ on $C$ and edges $vw\not\in E(C)$ and $xy\not\in E(C)$.
Let $P=vw\join\subpath{C}{w,x}\join xy$ and $Q=vw\join\subpath{\reverse{C}}{w,x}\join xy$.
\begin{enumerate}
\item \label{lem:repr:rot_paths_cycle:non-essential} If $C$ is a
  non-essential cycle and $vw$ and $xy$ lie in the exterior of $C$, or
\item\label{lem:repr:rot_paths_cycle:essential} if $C$ is an essential
  cycle and $vw$ lies in the exterior of $C$ and $xy$ in the interior,
  then it is $\rot(P)=\rot(Q)$.
\end{enumerate}
}

\begin{lemma}\label{lem:repr:rot_paths_cycle}
\lemmaRotPathsCycle
\end{lemma}

\begin{proof}
To prove the equality, we show in the following that $\rot(P)-\rot(Q)=0$. We first cut the paths in two parts each and analyze them separately.
To this end we choose arbitrary edges $e$ on $\subpath{P}{w,x}$ and $e'$ on $\subpath{Q}{w,x}$ as shown in Fig.~\ref{fig:repr:paths_cycle}. By Observation~\ref{obs:rot_splitting_path}, it is
\begin{align}
\rot(P)-\rot(Q) =
&\rot(\subpath{P}{v,e}) - \rot(\subpath{Q}{v, e'})
+ \rot(\subpath{P}{e, y}) - \rot(\subpath{Q}{e', y}). \label{eqn:repr:rot_paths_cycle:split}
\end{align}

By definition it is $\subpath{P}{v,e}=vw\join \subpath{C}{w,e}$ and $\subpath{Q}{v,e'}=vw\join \subpath{\reverse{C}}{w,e'}$. Therefore, we can calculate that difference of the rotations as follows:
\begin{align}
\rot(\subpath{P}{v,e})-\rot(\subpath{Q}{v,e'})
&=\rot(\subpath{P}{v,e})+\rot(\subpath{\reverse{Q}}{v,e'}) \nonumber\\
&=\rot(vw\join \subpath{C}{w,e}) + \rot(\subpath{C}{\reverse{e'},w}\join wv) \nonumber\\
&=\rot(\subpath{C}{\reverse{e'}, e})-2 \label{eqn:repr:rot_paths_cycle:first}
\end{align}
The last equality follows from Observation~\ref{obs:pre:rot_path_detour}, since $vw$ lies in the exterior of $C$ and therefore locally to the left of $\subpath{C}{\reverse{e'}, e}$.
In a similar spirit one obtains an equation for the second part:
\begin{equation}
\rot(\subpath{P}{e,y})-\rot(\subpath{Q}{e',y})=\rot(\subpath{C}{e,\reverse{e'}}) + c \label{eqn:repr:rot_paths_cycle:second}
\end{equation}
Here, $c$ represents a constant, which depends on whether $xy$ lies in the exterior of $C$ (Case~\ref{lem:repr:rot_paths_cycle:non-essential}; $c=-2$) or in the interior (Case~\ref{lem:repr:rot_paths_cycle:essential}; $c=2$).

Substituting Equations~\ref{eqn:repr:rot_paths_cycle:first} and \ref{eqn:repr:rot_paths_cycle:second} in Equation~\ref{eqn:repr:rot_paths_cycle:split}, we get
\begin{equation}
\rot(P)-\rot(Q) = \rot(\subpath{C}{\reverse{e'}, e}) + \rot(\subpath{C}{e,\reverse{e'}}) - 2 + c.
\end{equation}
Note that the paths in the equation above together form the cycle~$C$. Hence, it is $\rot(\subpath{C}{\reverse{e'}, e}) + \rot(\subpath{C}{e,\reverse{e'}}) = \rot(C)$ and the equation simplifies to
\begin{equation}\label{eqn:repr:rot_paths_cycle:final}
\rot(P)-\rot(Q) = \rot(C) - 2 + c.
\end{equation}

In Case \ref{lem:repr:rot_paths_cycle:non-essential} the cycle $C$ is
non-essential and therefore it is $\rot(C)=4$. Moreover, we have
$c=-2$.  Plugging these values in
Equation~\ref{eqn:repr:rot_paths_cycle:final}, we get
$\rot(P)-\rot(Q)=0$.

Similarly, in Case \ref{lem:repr:rot_paths_cycle:essential} it is
$c=2$ and $\rot(C)=0$, since $C$ is essential. Therefore,
Equation~\ref{eqn:repr:rot_paths_cycle:final} again gives the desired
result.
\end{proof}

Applying this result repeatedly to a pair of paths, we show that the labels of an edge~$e$ on an essential cycle $C$ induced by two elementary paths $P$ and $Q$ are the same, i.e., $\ell^P_C(e)=\ell^Q_C(e)$.
We even prove a slightly more general result, where the paths do not need to start at the reference edge but stay between two essential cycles $C_1$ and $C_2$; see  Fig.~\ref{fig:repr:rotation_paths_to_cycle}.

\newcommand{\lemRotationPathsToCycle}{
Let $C_1$ and $C_2$ be two disjunct essential cycles such that $C_2$ lies in the interior of $C_1$.
Fix an edge $rs$ on $C_1$ and an edge $e=tu$ on $C_2$. Let $P$ and $Q$ be two paths starting at $s$ and ending at $t$ such that both paths lie in the interior of $C_1$ and in the exterior of $C_2$. Then $\rot(rs\join P\join tu)=\rot(rs\join Q\join tu)$.
}
\begin{lemma}\label{lem:repr:rotation_paths_to_cycle}
\lemRotationPathsToCycle
\end{lemma}

\begin{figure}
  \centering
  \includegraphics{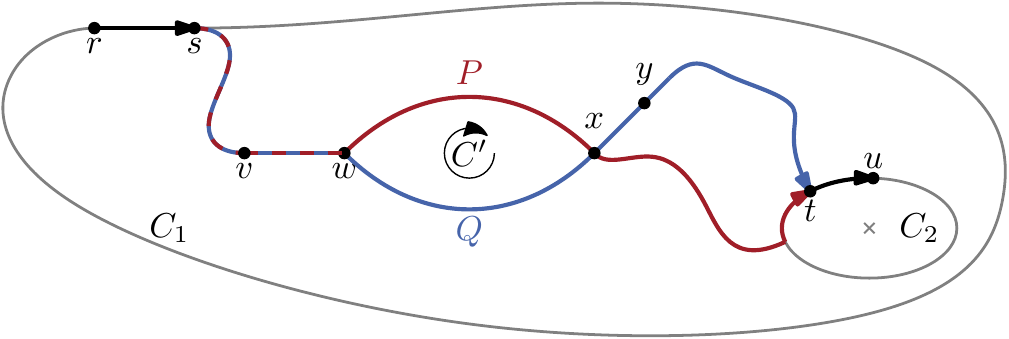}
  \caption{Two paths $P$ and $Q$ from $s$ to $t$, which lie between the essential cycles $C_1$ and $C_2$. Here, the cycle $C'$ formed by $\subpath{P}{w,x}$ and $\subpath{\reverse{Q}}{x,w}$ is non-essential.}
  \label{fig:repr:rotation_paths_to_cycle}
\end{figure}

  \begin{proof}
    Let $k$ be the number of directed edges on $P$ that do not lie on
    $Q$. We prove the equivalence of $P$ and $Q$ by induction on $k$.
    If $k=0$, $Q$ contains $P$ completely. Since both paths have the
    same start and end points, $P$ and $Q$ are equal. Hence, the claim
    follows immediately.

    If $k>0$, there is a first edge $vw$ on $P$ such that the
    following edge does not lie on $Q$. Let~$x$ be the first vertex on
    $P$ after $w$ that lies on $Q$ and let $y$ be the vertex on $Q+tu$
    that follows $x$ immediately.  This situation is illustrated in
    Fig.~\ref{fig:repr:rotation_paths_to_cycle}.  As both $P$ and $Q$
    end at~$t$, these vertices always exist.  Consider the cycle
    $C'=\subpath{P}{w,x}\join{\subpath{\reverse{Q}}{x,w}}$. We may
    assume without loss of generality that the edges of $C'$ are
    directed such that the outer face lies in the exterior of $C'$
    (otherwise we may reverse the edges and work with $\reverse{C'}$
    instead).  Note that $\subpath{Q}{x,t}$ cannot intersect $C'$
    (except for $x$), since the interior of $\subpath{P}{w,x}$ does
    not intersect $Q$ by the definition of $x$. Therefore, $xy$ lies
    on the same side of $C'$ as $C_2$.

    If $C'$ is non-essential, it does not contain the central face,
    and hence, $C'$ cannot contain $C_2$ in its interior. Thus, both
    $vw$ and $xy$ lie in the exterior of $C'$.  By
    Lemma~\ref{lem:repr:rot_paths_cycle}, we have $\rot(vw\join
    \subpath{P}{w,x}\join xy) = \rot(vw\join\subpath{Q}{w,x}\join
    xy)$.

    If $C'$ is essential, $C_2$ lies in the interior of $C'$. Hence,
    $vw$ lies in the exterior of $C'$ and $xy$ in its interior or
    boundary. Lemma~\ref{lem:repr:rot_paths_cycle} states that
    $\rot(vw\join \subpath{P}{w,x}\join xy) =
    \rot(vw\join\subpath{Q}{w,x}\join xy)$.

    In both cases it follows from $\subpath{P}{s,w}=\subpath{Q}{s,w}$
    that $\rot(rs\join\subpath{P}{s,x}\join
    xy)=\rot(rs\join\subpath{Q}{s,x}\join xy)$.  For
    $Q'=\subpath{P}{s,x}\join\subpath{Q}{x,t}$ it therefore holds that
    $\rot(rs\join Q\join tu)=\rot(rs\join Q'\join tu)$.  As $Q'$
    includes the part of $P$ between $w$ and $x$ it misses fewer edges
    from $P$ than $Q$ does. Hence, the induction hypothesis implies
    $\rot(rs\join P\join tu)=\rot(rs\join Q'\join tu)$, and thus
    $\rot(rs\join P\join tu)=\rot(rs\join Q\join tu)$.
  \end{proof}

Combining Lemma~\ref{lem:repr:rotation_paths_to_cycle} and Lemma~\ref{lem:repr:equal_label_all_equal} we show that labelings induced by elementary paths are equal.

\newcommand{\lemEquivalencePseudoElementaryPaths}{
Let $C$ be an essential cycle in $G$ and let $e^\star=rs$. If $P$ and $Q$ are paths from $s$ to vertices on $C$ such that $P$ and $Q$ lie in the exterior of $C$, they are equivalent.
}

\begin{lemma}\label{lem:repr:equivalence_pseudo_elementary_paths}
 \lemEquivalencePseudoElementaryPaths
\end{lemma}

\begin{proof}
  First assume that one of the paths, say $Q$, is elementary. Let $p$
  and $q$ be the endpoints of $P$ and $Q$, respectively. Let $v$ be
  the vertex following $p$ on $C$ when $C$ is directed such that the
  central face lies in its interior.  We define
  $Q'=Q\join\subpath{C}{q,p}$.  It is easy to verify that both $P$ and
  $Q'$ are paths and lie in the exterior of $C$.  Therefore, it
  follows from Lemma~\ref{lem:repr:rotation_paths_to_cycle} that \[
  \ell^P_C(pv)=\rot(e^\star\join P\join pv)=\rot(e^\star\join Q'\join
  pv)=\ell^Q_C(pv).
\]
Thus, $P$ and $Q$ are equivalent by
Lemma~\ref{lem:repr:equal_label_all_equal}.

If neither $P$ nor $Q$ are elementary, choose any elementary path~$R$ to a vertex on $C$. The argument above shows that $P\equiv_C R$ and $R\equiv_C Q$, and thus $P\equiv_C Q$.
\end{proof}

In the remainder of this section all labelings are induced by elementary paths.
By Lemma~\ref{lem:repr:equivalence_pseudo_elementary_paths}, the labelings are independent of the choice of the elementary path. Hence, we drop the superscript $P$ and write $\ell_C(e)$ for the labeling of an edge $e$ on an essential cycle~$C$.

\begin{figure}[bt]
  \centering
  \subfloat[]{ \label{fig:repr:intersection_cycles-complete}
     \includegraphics{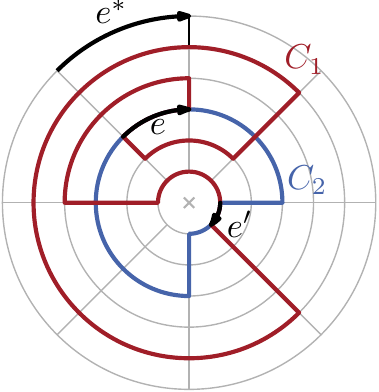}
  } \hspace{1cm}
  \subfloat[]{ \label{fig:repr:intersection_cycles-C1}
             \includegraphics{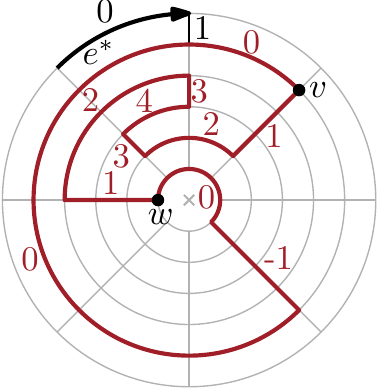}
  }
  \caption{\protect\subref{fig:repr:intersection_cycles-complete} Two cycles 
  $C_1$ and $C_2$ may have both common edges with different labels 
  ($\ell_{C_1}(e)=4\neq 0 = \ell_{C_2}(e)$) and ones with equal labels 
  ($\ell_{C_1}(e')=\ell_{C_2}(e')=0$).
  \protect\subref{fig:repr:intersection_cycles-C1} All labels of 
  $\subpath{C_1}{v,w}$ are positive, implying that $C_1$ goes down. Note that 
  not all edges of $\subpath{C_1}{v,w}$ point downwards.}
  \label{fig:repr:intersection_cycles}
\end{figure}

If an edge $e$ lies on two simple, essential cycles $C_1$ and $C_2$,
the labels $\ell_{C_1}(e)$ and $\ell_{C_2}(e)$ may not be equal in
general.  In Fig.~\ref{fig:repr:intersection_cycles-complete} the labels of the
edge $e$ are different for $C_1$ and $C_2$ respectively, but $e'$ has label~$0$ on both cycles. Note
that $e'$ is incident to the central face of $C_1+C_2$.  We now show
that this is a sufficient condition for the equality of the labels.

\newcommand{\lemEqualLabelsAtIntersection}{
  Let $C_1$ and $C_2$ be two essential cycles and let $H=C_1+C_2$ be
  the subgraph of $G$ formed by these two cycles. Let $v$ be a common
  vertex of $C_1$ and $C_2$ on the central face of $H$ and consider
  the edge $vw$ on $C_1$.  Denote the vertices before $v$ on $C_1$ and
  $C_2$ by $u_1$ and $u_2$, respectively.  Then $\ell_{C_1}(u_1v) +
  \rot(u_1vw) = \ell_{C_2}(u_2v) + \rot(u_2vw)$.

Further, if $vw$ belongs to both $C_1$ and $C_2$, the labels of $e$ are equal, i.e., $\ell_{C_1}(vw)=\ell_{C_2}(vw)$.
}

\begin{lemma}\label{lem:repr:equal_labels_at_intersection}
\lemEqualLabelsAtIntersection
\end{lemma}

\begin{proof}
  First assume that not only $v$ but the whole edge $vw$ lies on the
  central face of $H$.  Let $C$ be the cycle bounding the central face
  of $C_1+C_2$. Let $P$ and $Q$ be elementary paths from the
  endpoint~$s$ of the reference edge to vertices $t_1$ and $t_2$ on
  $C_1$ and $C_2$, respectively.  Define
  $P'=P\join\subpath{C_1}{t_1,v}$ and $Q'=Q\join\subpath{C_2}{t_2,v}$.
  The paths $P'$ and $Q'$ lie in the exterior of $C$.  Thus,
  Lemma~\ref{lem:repr:rotation_paths_to_cycle} can be applied to $P'$
  and $Q'$:
\begin{align*}
  \rot(e^\star\join P') + \rot(u_1vw) &= \rot(e^\star\join P'\join vw) \\
  &=\rot(e^\star\join Q'\join vw) = \rot(e^\star\join Q')+\rot(u_2vw).
\end{align*}
By the definition of labelings it is
$\ell_{C_1}(u_1v)=\rot(e^\star\join P')$ and
$\ell_{C_2}(u_2v)=\rot(e^\star\join Q')$. Substitution into the
previous equation give the desired result:
\[
\ell_{C_1}(u_1v) + \rot(u_1vw) = \ell_{C_2}(u_2v) + \rot(u_2vw).
\]

If $vw$ does not lie on the central face, then the edge $vw_2$ on
$C_2$ does lie on the central face, where $w_2$ denotes the vertex on
$C_2$ after $v$. By the argument above we have
\[
\ell_{C_1}(u_1v) + \rot(u_1vw_2) = \ell_{C_2}(u_2v) + \rot(u_2vw_2).
\]
Since $vw$ lies locally to the left of both $u_1vw_2$ and $u_2vw_2$,
it is $\rot(u_ivw)=\rot(u_ivw_2)-\alpha$ for $i=1,2$ and the same
constant $\alpha$, which is either $1$ or $2$. Hence, we get
\[
\ell_{C_1}(u_1v) + \rot(u_1vw) + \alpha = \ell_{C_2}(u_2v) + \rot(u_2vw) + \alpha.
\]
Canceling $\alpha$ on both sides gives the desired result.

If $vw$ lies on both $C_1$ and $C_2$,
Observation~\ref{obs:repr:label_difference} shows $\ell_{C_1}(u_1v) +
\rot(u_1vw) = \ell_{C_1}(vw)$ and $\ell_{C_2}(u_2v) + \rot(u_2vw) =
\ell_{C_2}(vw)$.  Hence, the result above can be simplified to
$\ell_{C_1}(vw) = \ell_{C_2}(vw)$.
\end{proof}

Intuitively, positive labels can often be interpreted as going downwards and negative labels as going upwards.
In Fig.~\ref{fig:repr:intersection_cycles-C1} all edges of $\subpath{C_1}{v,w}$ have positive labels and in total the $y$-coordinate decreases along this path, i.e., the $y$-coordinate of $v$ is greater than the $y$-coordinate of $w$.

Using this intuition, we expect that a path~$P$ going down cannot
intersect a path~$Q$ going up such that $P$ starts below $Q$ and ends
above it.  In Lemma~\ref{lem:repr:illegal_intersection}, we show that
this assumption is correct if we restrict ourselves to intersections
on the central face.

\newcommand{\lemIllegalIntersection}{
Let $C$ and $C'$ be two simple, essential cycles in $G$ sharing at least one vertex.
Let $H$ be the subgraph of $G$ composed of the two cycles $C$ and $C'$ and denote the central face of $H$ by $f$. Take subpaths $uvw$ of $C$ and $u'vw'$ of $C'$ with the same vertex~$v$ in the middle, which also lies on $f$.
\begin{enumerate}
\item\label{lem:repr:illegal_intersection-in} If $\ell_C(uv)\geq
  0$ and $\ell_{C'}(u'v)\leq 0$, $u'v$ lies in the interior of $C$.
\item\label{lem:repr:illegal_intersection-out} If $\ell_C(vw)\geq 0$
  and $\ell_{C'}(vw')\leq 0$, $vw'$ lies in the exterior of $C$.
\end{enumerate}
}

\begin{lemma}\label{lem:repr:illegal_intersection}
  \lemIllegalIntersection
\end{lemma}

\begin{figure}[bt]
 \centering
 \subfloat[The labels of the incoming edges satisfy $\ell_C(uv)\geq 0$ and $\ell_{C'}(u'v)\leq 0$. The edges $vw$ and $vw'$ could be exchanged.]{ \label{fig:repr:illegal_intersection-in}\includegraphics{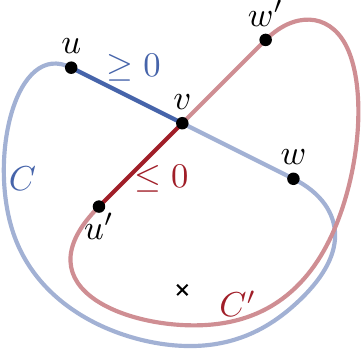} }
 \hspace{1cm}
 \subfloat[The labels of the outgoing edges satisfy $\ell_C(vw)\geq 0$ and $\ell_{C'}(vw')\leq 0$. The edges $uv$ and $u'v$ could be exchanged.]{ \label{fig:repr:illegal_intersection-out}\includegraphics{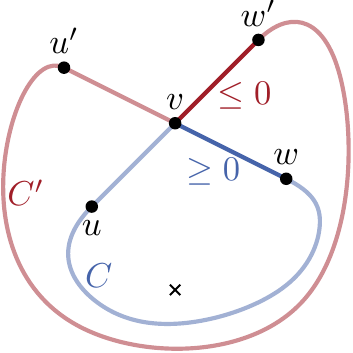} }
 \caption{Possible intersection of two cycles~$C$ and $C'$ at $v$. }
 \label{fig:repr:illegal_intersection}
\end{figure}

\begin{proof}
  Case \ref{lem:repr:illegal_intersection-in}. Since the central face~$f$
  lies in the interior of both $C$ and $C'$ and $v$ on the boundary of
  $f$, one of the edges $vw$ and $vw'$ lies on $f$. We denote this
  edge by $vx$ and it is either $x=w$ (as in
  Fig.~\ref{fig:repr:illegal_intersection-in}) or $x=w'$.  By the
  Lemma~\ref{lem:repr:equal_labels_at_intersection}, we have $
  \ell_C(uv) + \rot(uvx) = \ell_{C'}(u'v) + \rot(u'vx) $.  Applying
  $\ell_C(uv)\geq 0$ and $\ell_{C'}(u'v)\leq 0$, we obtain $\rot(uvx)
  \leq \rot(u'vx)$. Therefore, $u'v$ lies to the right of or on $uvx$
  and thus in the interior of $C$.

  Case \ref{lem:repr:illegal_intersection-out}. Assume that $vw'$ does
  not lie in the exterior of $C$, that is, $vw'$ lies locally to the
  right of $uvw$. Therefore, $vw'$ lies on $f$, and
  Lemma~\ref{lem:repr:equal_labels_at_intersection} implies that $ 0
  \geq \ell_{C'}(vw') = \ell_C(uv) + \rot(uvw') $ Since $vw'$ lies
  locally to the right of $uvw$, it is $\rot(uvw) < \rot(uvw')$ and
  therefore $0 \geq \ell_C(uv) + \rot(uvw') > \ell_C(uv)+\rot(uvw)
  =\ell_C(vw)$ Here, the last equality follows from
  Observation~\ref{obs:repr:label_difference}. But this contradicts
  the assumption $\ell_C(vw)\geq 0$. Hence, $vw'$ must lie in the
  exterior of $C$.
\end{proof}

For the correctnes proof in Section~\ref{sec:rectangulation}, a crucial insight is that for cycles using an edge which is part of a face, we can find an alternative cycle without this edge in a way that preserves labels on the common subpath:

\newcommand{\lemExistenceC}{ If an edge $e$ belongs to both a simple
  essential cycle~$C$ and a regular face~$f'$ with $C \neq f'$, then there
  is a simple essential cycle~$C'$ not containing $e$ such that $C'$
  can be decomposed into two paths $P$ and $Q$, where $P$ or
  $\reverse{P}$ lies on $f'$ and $Q=C\cap C'$.  }

\begin{lemma}\label{lem:rect:existence_C'}
\lemExistenceC
\end{lemma}

\begin{figure}[bt]
\centering
  \subfloat[The cycle bounding the outer face is $C'$.]{
    \centering
    \includegraphics{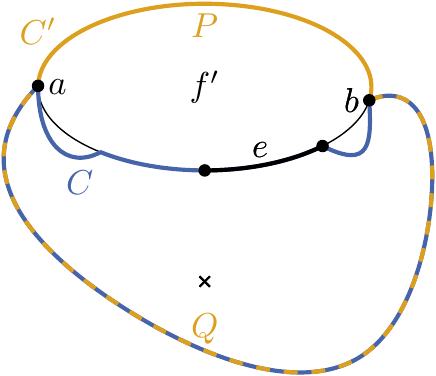}
    \label{fig:rect:existence_C'-outer}
  }
    \hspace{1ex}
  \subfloat[The cycle bounding the central face is $C'$.] {
    \centering
    \includegraphics{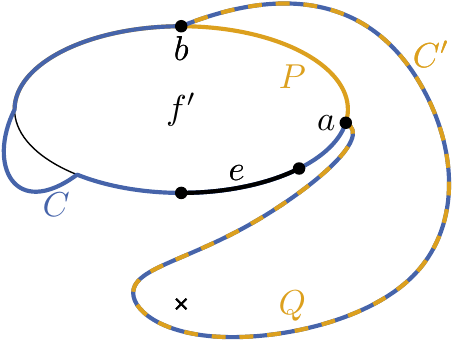}
    \label{fig:rect:existence_C'-center}
  }

  \caption{The edge $e$ cannot lie on both the outer and the central face, which is marked by a cross.
  (a) $e$ does not lie on the outer face, and hence the cycle bounding this face is defined as $C'$. (b) $C'$ is the cycle bounding the central face. In both cases $C'$ can be subdivided in two paths $P$ and $Q$ on $C$ and $f'$, respectively.
  Here, these paths are separated by the vertices $a$ and $b$.}
  \label{fig:rect:existence_C'}
\end{figure}

\begin{proof}
Consider the graph $H=C+f'$ composed of the essential cycle $C$ and
the regular face~$f'$. In $H$ the edge $e$ cannot lie on both the
outer and the central face.  If $e$ does not lie on the outer face, we
define $C'$ as the cycle bounding the outer face but directed such
that it contains the center in its interior (see
Fig.~\ref{fig:rect:existence_C'-outer}).  Otherwise, $C'$ denotes the
cycle bounding the central, which is illustrated in
Fig.~\ref{fig:rect:existence_C'-center}.

Since $C$ lies in the exterior of $f'$, the intersection of $C$ with
$C'$ forms one contiguous path~$Q$.  Setting $P=C-Q$ yields a path
that lies completely on $f'$ (it is possible though that $P$ and $f'$
are directed differently).
\end{proof}

Using this lemma, we construct an essential cycle $C'$ without the new edge $uz$ from an essential cycle~$C$ including $uz$.
Also, $C$ and $C'$ have a common path~$P$ lying on the central face of $H = C + f'$. Thus, Lemma~\ref{lem:repr:equal_labels_at_intersection} implies labelings of $C$ and $C'$ being equal on $P$.
\begin{corollary}\label{cor:rect:existence_C'_intersection}
For essential cycles $C$, $C'$ and the path $P=C\cap C'$ from Lemma~\ref{lem:rect:existence_C'}, it is $\ell_C(e)=\ell_{C'}(e)$ for all edges $e$ on $P$.
\end{corollary}

For two intersecting essential cycles where all labels on one cycle are 0, we 
also observe:

\begin{SCfigure}[50][tb]
 \centering
 \includegraphics{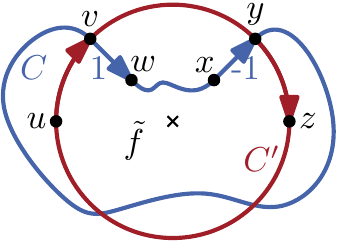}
 \caption{The situation of Lemma~\ref{lem:rect:two_cycles_horizontal}. All 
 edges of $C'$ are labeled with 0. In this situation there are edges on $C$ 
 with labels $-1$ and $1$. Hence, $C$ is neither increasing nor 
 decreasing.}
 \label{fig:rect:horizontal_cycle}
\end{SCfigure}

\newcommand{\lemTwoCyclesHorizontal}{
Let $C$ and $C'$ be two essential cycles that have at least one common vertex. If all edges on $C'$ are labeled with $0$, $C$ is neither increasing nor decreasing.
}

\begin{lemma}\label{lem:rect:two_cycles_horizontal}
\lemTwoCyclesHorizontal
\end{lemma}

\begin{proof}
The situation is illustrated in Fig.~\ref{fig:rect:horizontal_cycle}.
If the two cycles are equal, the claim clearly holds.
Otherwise, we show that one can find two edges on $C$ such that the labels of these edges have opposite signs.

Let $vw$ be an edge of $C$ but not of $C'$ such that $v$ lies on the central face~$\tilde f$ of $H=C+C'$, and denote the vertex before $v$ on $C'$ by $u$.
By Lemma~\ref{lem:repr:equal_labels_at_intersection} it is
\begin{equation}\label{eqn:rect:two_cycles_horizontal:vw}
\ell_{C}(vw)=\ell_{C'}(uv)+\rot(uvw) = \rot(uvw).
\end{equation}
The second equality follows from the assumption that $\ell_{C'}(uv)=0$.
Let $y$ be the first common vertex of $C$ and $C'$ on the central face~$\tilde f$ after $v$.
That is, $\subpath{\tilde f}{v,y}$ is a part of one of the cycles $C$ and $C'$ and intersects the other cycle only at $v$ and $y$. We denote the vertex on $C$ before $y$ by $x$ and the vertex after $y$ on $C'$ by $z$.
Again by Lemma~\ref{lem:repr:equal_labels_at_intersection} we have
\begin{equation}\label{eqn:rect:two_cycles_horizontal:xy}
\ell_{C}(xy) = \ell_{C'}(yz)-\rot(xyz)=-\rot(xyz).
\end{equation}
By construction $vw$ and $xy$ lie on the same side of $C'$. Hence,
$uvw$ and $xyz$ both make a right turn if $vw$ and $xy$ lie in the
interior of $C'$ and a left turn otherwise.  Thus, it is
$\rot(uvw)=\rot(xyz)\neq 0$, and
Equations~\ref{eqn:rect:two_cycles_horizontal:vw} and
\ref{eqn:rect:two_cycles_horizontal:xy} imply that $\ell_C(uv)$ and
$\ell_C(xy)$ have opposite signs.  Hence, $C$ is neither in-\ nor
decreasing.
\end{proof}

\section{Characterization of Rectangular Graphs}
\label{sec:characterization-rect}

In this section, we prove Theorem~\ref{thm:repr:characterization} for rectangular graphs. A \emph{rectangular graph} is a graph together with an ortho-radial representation in which every face is a rectangle.
We define two flow networks that assign consistent lengths to the
graph's edges, one for the vertical and one
for the horizontal edges.  These networks are straightforward adaptions of the
networks used for drawing rectangular graphs in the
plane~\cite{bett-gdavg-99}.  In the following, \emph{vertex}
and \emph{edge} refer to the vertices and edges of the graph~$G$,
whereas \emph{node} and \emph{arc} are used for the flow networks.

\begin{figure}[t]
\centering
  \subfloat[$N_\text{ver}$]{ 
  \label{fig:draw:flows-ver}\includegraphics[scale=1]{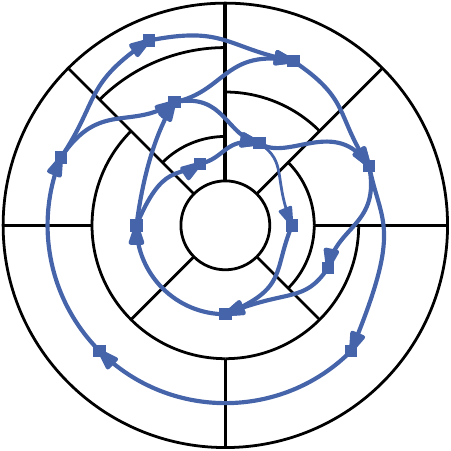}}
  ~
  \subfloat[$N_\text{rad}$]{ 
  \label{fig:draw:flows-hor}\includegraphics[scale=1]{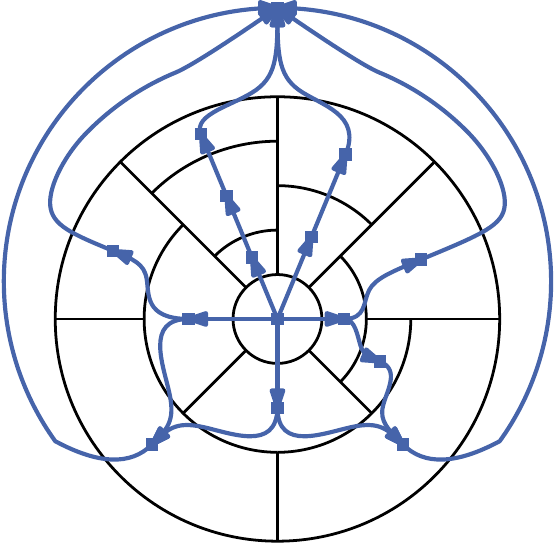}}

  \caption{Flow networks $N_\text{ver}$ and $N_\text{rad}$ for an example 
  graph~$G$. For simplicity, the edge from the outer to the central face in 
  $N_\text{rad}$ is omitted.}
  \label{fig:draw:flows}
\end{figure}

The network $N_\text{ver}=(F_\text{ver}, A_\text{ver})$ with nodes~$F_\text{ver}$ and arcs~$A_\text{ver}$ for the vertical edges contains one node for each face of $G$ except for the central and the outer face. All nodes have a demand of~0.
For each vertical edge $e$ in $G$, which we assume to be directed upwards, there is an arc $a_e=fg$ in $N_\text{ver}$, where $f$ is the face to the left of $e$ and $g$ the one to its right. The flow on $fg$ has the lower bound $l(fg)=1$ and upper bound $u(fg)=\infty$. An example of this flow network is shown in Fig.~\ref{fig:draw:flows-ver}.

To obtain a drawing from a flow in $N_\text{ver}$, we set the length of a vertical edge~$e$ to the flow on $a_e$. The conservation of the flow at each node $f$ ensures that the two vertical sides of the face~$f$ have the same length.

Similarly, the network~$N_\text{rad}$ assigns lengths to the radial edges.
There is a node for all faces of $G$, and an arc $a_e=fg$ for every horizontal edge $e$ in $G$, which we assume to be oriented clockwise.
Additionally, $N_\text{rad}$ includes one arc from the outer to the central face.
Again, all edges require a minimum flow of 1 and have infinite capacity. The 
demand of all nodes is 0. Fig.~\ref{fig:draw:flows-hor} shows an example of 
such a flow network.

\begin{lemma}
  \label{thm:rect:flows-to-drawing}
  A pair of valid flows in $N_\text{rad}$ and $N_\text{ver}$ bijectively corresponds to a valid ortho-radial drawing of $G$ respecting $\Gamma$.
\end{lemma}
\begin{proof}
Given a feasible flow $\varphi$ in $N_\text{ver}$, we set the length of each 
vertical edge~$e$ of $G$ to the flow $\varphi(a_e)$ on the arc $a_e$ that
crosses $e$. For each face $f$ of $G$, the total length of its left
side is equal to the total amount of flow entering $f$.  Similarly,
the length of the right side is equal to the amount of flow leaving
$f$. As the flow is preserved at all nodes of $N_\text{ver}$, the left
and right sides of $f$ have the same length.  Similarly, one obtains
the length of the radial edges from a flow in $N_\text{rad}$.

Similarly, given a drawing of $\Gamma$, we can extract a flow in the networks. 
Since the opposing sides of the rectangles have the same length, the flow is 
preserved at the nodes. Moreover, all sides have length at least $1$ and 
therefore the minimum flow on all arcs is guaranteed.
\end{proof}

Using this equivalence of drawings and feasible flows, we show the 
characterization of rectangular graphs.

\begin{theorem}\label{thm:draw:rectangle_drawing}
  Let $G=(V,E)$ be a plane graph with an ortho-radial representation~$\Gamma$ 
  satisfying Conditions~\ref{cond:repr:sum_of_angles} and 
  \ref{cond:repr:rotation_faces} of 
  Definition~\ref{def:repr:valid_representation} such that all faces are 
  rectangles.
  Let $N_\text{rad}$ and $N_\text{ver}$ be the flow networks as defined above.
  The following statements are equivalent:
  \begin{enumerate}[label=(\roman*)]
  \item\label{item:draw:rectangle_drawing:drawing} There is an ortho-radial 
  drawing of $\Gamma$.
  \item\label{item:draw:rectangle_drawing:set} No $S\subseteq F_\text{ver}$ 
  exists such that there is an arc from $F_\text{ver}\setminus S$ to $S$ in 
  $N_\text{ver}$ but not vice versa.
  \item\label{item:draw:rectangle_drawing:valid} $\Gamma$ is valid.
  \end{enumerate}  
\end{theorem}

\begin{proof}
  \ref{item:draw:rectangle_drawing:drawing} $\Rightarrow$
  \ref{item:draw:rectangle_drawing:valid}: Let $\Delta$ be an
  ortho-radial drawing of $G$ preserving the embedding described by
  $\Gamma$ and let $C$ be a simple, essential cycle directed such that
  the center lies in its interior.  Our goal is to show that $C$ is
  neither an increasing nor a decreasing cycle.  To this end, we
  construct a path $P$ from the reference edge to a vertex on $C$ such
  that the labeling of $C$ induced by $P$ attains both positive and
  negative values.

  In $\Delta$ either all vertices of $C$ have the same $y$-coordinate,
  or there is a maximal subpath $Q$ of $C$ whose vertices all have the
  maximum $y$-coordinate among all vertices of $C$.  In the first
  case, we may choose the endpoint~$v$ of the path~$P$ arbitrarily,
  whereas in the second case we select the first vertex of $Q$ as $v$.

\begin{figure}[bt]
  \centering
  \includegraphics[scale=0.9]{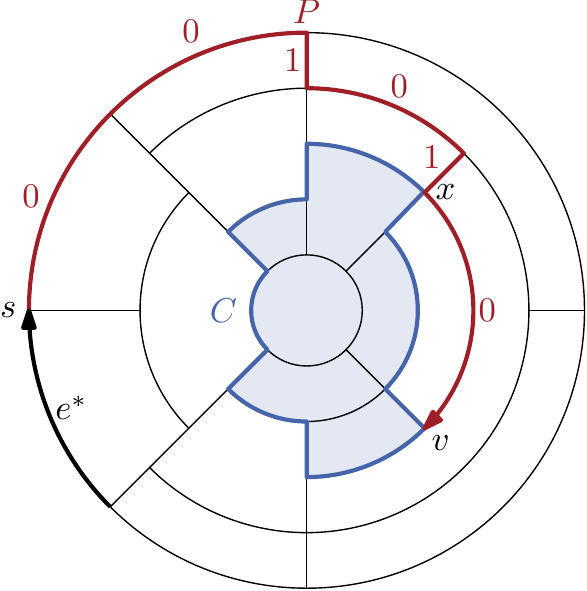}
  \caption{The path $P$ from $s$ to $v$---constructed backwards by going only 
  up or left---does not intersect the interior of $C$. The rotations of the 
  edges on $P$ relative to $e^\star$ are 0 or 1.}
  \label{fig:draw:drawing_to_representation}
\end{figure}

We construct the path $P$ backwards (i.e., the construction yields
$\reverse{P}$) as follows: Starting at $v$ we choose the edge going
upwards from $v$, if it exists, or the one leading left. Since all
faces of $G$ are rectangles, at least one of these always exists. This
procedure is repeated until the endpoint~$s$ of the reference edge is
reached.  An example of the constructed path is shown in
Fig.~\ref{fig:draw:drawing_to_representation}

To show that this algorithm terminates, we assume that this was not
the case.  As $G$ is finite, there must be a first time a vertex~$w$
is visited twice. Hence, there is a cycle~$C'$ in $\Delta$ containing
$w$ that contains only edges going left or up.  As all drawable
essential cycles with edges leading upwards must also have edges that
go down~\cite{hht-orthoradial-09}, all edges of $C'$ are
horizontal. By construction, there is no edge incident to a vertex of
$C'$ that leads upwards. The only cycle with this property, however,
is the one enclosing the outer face, because $G$ is connected. But
this cycle contains the reference, and therefore the algorithm halts.

This not only shows that the construction of $P$ ends, but also that
$P$ is a path (i.e., the construction does not visit a vertex twice).
However, $P$ might be not elementary, since it may intersect $C$
multiple times (e.g.\ in Fig.~\ref{fig:draw:drawing_to_representation}
$P$ contains two vertices of $C$: $v$ and $x$).  But $v$ has the
smallest $y$-coordinate among all vertices of $P$ and the largest
among those on $C$.
Thus, no part of $P$ lies inside $C$.  Hence,
Lemma~\ref{lem:repr:equivalence_pseudo_elementary_paths} guarantees
that the labeling $\ell^P_C$ induced by $P$ coincides with the
labeling $\ell_C$ induced by any elementary path.  To show that
$\Gamma$ satisfies Condition~\ref{cond:repr:labeling} of
Definition~\ref{def:repr:valid_representation}, we consider the
labeling $\ell^P_C$ induced by $P$.

Let $v'$ be the vertex following $v$ on $C$. By construction of $P$,
the labeling of $vv'$ induced by $P$ is $0$. If all edges of $C$ are
horizontal, this implies $\ell^P_C(e)=0$ for all edges $e$ of $C$,
which satisfies Condition~\ref{cond:repr:labeling}.

Otherwise, we claim that the edge $e_-=uv$ directly before $Q$ and the
edge $e_+=wx$ directly following $Q$ on $C$ have labels $-1$ and $+1$,
respectively.  Since all edges on $Q$ are horizontal and $e_-$ goes
down, we have $\rot(\subpath{C}{v,x})=1$ and therefore
$\ell^P_C(e_+)=1$. Similarly, $\rot(uvv')=1$ implies that
$\ell^P_C(uv)=\ell^P_C(vv') - \rot(uvv')=-1$.

\ref{item:draw:rectangle_drawing:set} $\Rightarrow$
\ref{item:draw:rectangle_drawing:drawing}:
By Lemma~\ref{thm:rect:flows-to-drawing} the existence of a drawing is 
equivalent to the existence of feasible flows in $N_\text{rad}$ and 
$N_\text{ver}$.
 We claim that $N_\text{rad}$
always possesses a feasible flow. To construct a flow, note that
$N_\text{rad}$ without the arc from the outer face~$g$ to the central
face~$f$ is a directed acyclic graph with $f$ as its only source and
$g$ as its only sink.  For each arc $a\neq gf$ in $N_\text{rad}$ there
is a directed path~$P_a$ from $f$ to $g$ via $a$.  Adding the
arc~$gf$, we obtain the cycle~$C_a=P_a\join gf$. Letting one unit flow
along each of these cycles $C_a$ and adding all flows gives the
desired flow in $N_\text{rad}$. By construction, each arc~$a$ lies on
at least one cycle (namely~$C_a$), and hence, the minimum flow of $1$
on $a$ is satisfied.

To construct a feasible flow in $N_\text{ver}$ we again compose cycles
of flow.  In order to find a directed cycle using an arc~$fg$, we
define the set $S_g$ of all nodes $h$ for which there exists a
directed path from $g$ to $h$ in $N_\text{ver}$. By definition, there
is no arc from a vertex in $S_g$ to a vertex not in $S_g$.  As
$N_\text{ver}$ satisfies~\ref{item:draw:rectangle_drawing:set}, $S_g$
does not have any incoming arcs either. Hence, $f\in S_g$ and there is
a directed path from $g$ to $f$. Closing this path with the arc $fg$
results in a cycle of $N_\text{ver}$, which we denote by $C_{fg}$.

Repeating this process for all arcs of $N_\text{ver}$, we obtain the
set~$\mathcal{C}$ of the cycles~$C_a$ ($a\in A_\text{ver}$). We set
the flow $\varphi(a)$ on each arc~$a$ as the number of the cycles in
$\mathcal{C}$ containing~$a$. Since $a$ lies on $C_a$, the flow on $a$
is at least 1.  As $\varphi$ is the sum of unit flows along cycles,
the flow is conserved at all nodes, i.e., the sum of the flows on
incoming arcs of a node $f$ is equal to the flow on arcs leaving $f$.
Therefore, $\varphi$ is a feasible flow in $N_\text{ver}$.

\ref{item:draw:rectangle_drawing:valid} $\Rightarrow$
\ref{item:draw:rectangle_drawing:set}: Instead of proving this
implication directly, we show the contrapositive. That is, we assume
that there is a set $S\subset F_\text{ver}$ of nodes in $N_\text{ver}$ such 
that
$S$ has no outgoing but at least one incoming arc.  From this
assumption we derive that $\Gamma$ is not valid, as we find an
increasing or decreasing cycle.

Let $N_\text{ver}[S]$ denote the node-induced subgraph of $N_\text{ver}$ 
induced by the set $S$. Without loss of generality, $S$ can be chosen such 
that $N_\text{ver}[S]$ is weakly connected.
If $N_\text{ver}[S]$ is not weakly connected, at least one weakly-connected 
component of $N_\text{ver}[S]$ possesses an incoming arc and we can work with 
this component instead.

As each node of $S$ corresponds to a face of $G$, $S$ can also be
considered as a collection of faces of $G$. To distinguish the two
interpretations of $S$, we refer to this collection of faces by
$\mathcal{S}$.  Our goal is to show that the innermost or the
outermost boundary of $\mathcal{S}$ forms a cycle in the original
graph $G$ that has no valid labeling.

\begin{figure}[bt]
 \centering
 \includegraphics{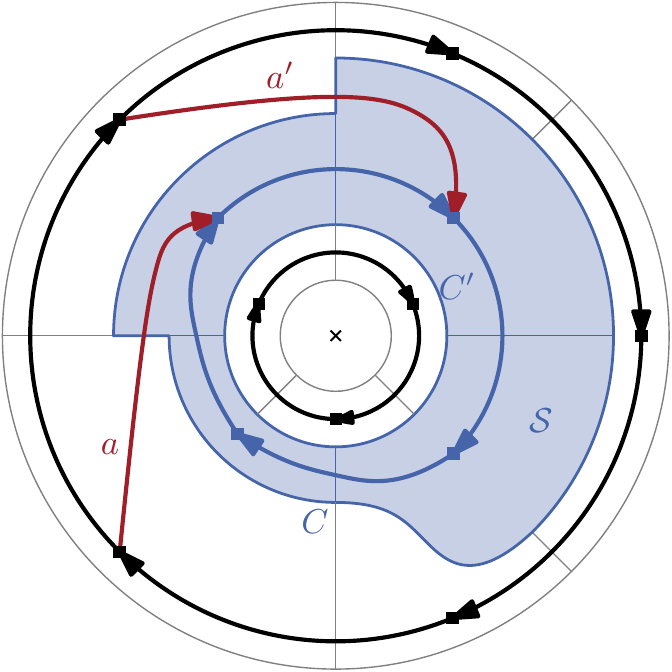}
 \caption{A set~$S$ of nodes in a graph~$G$ such that $N_\text{ver}[S]$ has no 
 outgoing but two incoming arcs~$a$ and $a'$.
 The set of faces~$\mathcal{S}$ corresponding to the nodes in $S$ are shaded 
 with blue.
 The outermost boundary of $\mathcal{S}$ forms an increasing cycle~$C$. The 
 edges on this cycle with label $-1$ are exactly those that are crossed by $a$ 
 or $a'$. All other edges on $C$ are labeled with 0.
 Note that the edge on $C$ at the bottom is curved because $G$ does not admit 
 an ortho-radial drawing.}
 \label{fig:draw:impossible_cycle}
\end{figure}

Fig.~\ref{fig:draw:impossible_cycle} shows an example of such a set
$S$ of nodes. Here, the arcs $a$ and $a'$ lead from a node outside of
$S$ to one in $S$. These arcs cross edges on the outer boundary of
$\mathcal{S}$, which point upwards.

Each node of $S$ has at least one outgoing arc, which must end at
another node of~$S$. Hence, $S$ contains a cycle, and therefore
$\mathcal{S}$ separates the outer and the central face of $G$.  Thus,
the cycle~$C$ that constitutes the outermost boundary of
$\mathcal{S}$, i.e., the smallest cycle containing all faces of
$\mathcal{S}$, is essential.  Similarly, we define $C'$ as the cycle
forming the innermost boundary of $\mathcal{S}$.

By assumption there is an arc~$a$ from a node entering $S$ from the
outside.  Since $N_\text{ver}[S]$ is weakly connected, $\mathcal{S}$
has no holes and the edge~$e$ crossed by $a$ either lies on $C$ or
$C'$. In the former case $e$ points up whereas in the latter case $e$
points down.  If $e$ lies on $C$ we prove in the following that $C$ is
an increasing cycle. Similarly one can show that $C'$ is a decreasing
cycle if $e$ lies on $C'$. We only present the argument for $C$ as the
one for $C'$ is similar.

We first observe that no edge on $C$ is directed downwards, since the
arc~$a_e$ corresponding to a downward edge~$e$ would lead from a node
in $S$ to a node outside of $S$.  To restrict the possible labels for
edges on $C$, we construct an elementary path~$P$ from the endpoint
$s$ of the reference edge to a vertex on $C$.  The construction is
similar to the one we used to show that the existence of a drawing
implies the validity of the representation, but this time the
construction works forwards starting at~$s$.  If the current vertex
lies on $C$, the path is completed. Otherwise, we append the edge
going down, if it exists, or the one that goes to the right.  As above
one can show that this procedure produces a path~$P$. It is even
elementary, because we stop when a vertex on $C$ is reached.

Let $v$ be the endpoint of $P$ and $w$ the vertex on $C$ following
$v$. By construction, $\rot(e^\star\join P)\in\{0,1\}$. Therefore,
$\ell_C(vw)=0$ if $vw$ points right or $\ell_C(vw)=-1$ if it points
up.  Since no edge on $C$ has a label that is congruent to 1 modulo 4
(i.e., the edge is downwards) and the labels of neighboring edges
differ by $-1$, $0$ or $1$, we obtain $\ell_C(e')\in\{-2,-1,0\}$ for
all edges $e'\in C$.  In particular, $\ell_C(e')\leq 0$. But the edge
$e$ crossed by the arc~$a$ points upwards and therefore
$\ell_C(e)=-1$. Hence, $C$ is an increasing cycle and $\Gamma$ is not
valid.
\end{proof}

By \cite{hht-orthoradial-09} an ortho-radial drawing of a graph satisfies 
Conditions~\ref{cond:repr:sum_of_angles} and~\ref{cond:repr:rotation_faces} of 
Definition~\ref{def:repr:valid_representation}. Therefore,
Theorem~\ref{thm:draw:rectangle_drawing} implies the characterization of 
ortho-radial drawings for rectangular graphs.

\begin{corollary}[Theorem~\ref{thm:repr:characterization} for Rectangular Graphs]\label{cor:draw:characterization}
  A rectangular 4-plane graph admits a bend-free ortho-radial drawing if and only if
  its ortho-radial representation is valid.
\end{corollary}

\section{Characterization of 4-Planar Graphs}\label{sec:rectangulation}

In the previous section we proved for rectangular graphs that there is
an ortho-radial drawing if and only if the ortho-radial representation
is valid. We extend this result to 4-planar graphs by reduction to the
rectangular case. In this section we provide a high-level overview of
the reduction; a detailed proof is found in Section~\ref{sec:rect:correctness}.

In Section~\ref{sec:rect:algorithm} we present an algorithm augmenting
$G$ such that all faces become rectangles. As we show in
Section~\ref{sec:rect:correctness} we then can apply
Corollary~\ref{cor:draw:characterization} to show the remaining
implication of Theorem~\ref{thm:repr:characterization}.

\subsection{Rectangulation Algorithm}\label{sec:rect:algorithm}
Given a graph $G$ and its ortho-radial representation $\Gamma$ as
input, we present a rectangulation algorithm that augments~$\Gamma$
producing a graph $G'$ with ortho-radial representation $\Gamma'$ such
that all faces except the central and outer face of $\Gamma'$ are
rectangles. Moreover, we ensure that the outer and the central face of
$\Gamma'$ make no turns.  

Let $u$ be a vertex that has a left turn on a face $f$ and $vw$ be an
edge on $f$. Let further $t$ be the vertex before $u$ on $f$. We
obtain an \emph{augmentation} $\Gamma^{u}_{vw}$ from $\Gamma$ by
splitting the edge~$vw$ into two edges $vz$ and $zw$ introducing a new
vertex~$z$.  Further, we insert the edge $uz$ in the interior of $f$
such that $uz$ points in the same direction as $tu$. For an
illustration of an augmentation see Fig.~\ref{fig:rect:insertion}.

\begin{figure}[tb]
 \centering
 \includegraphics{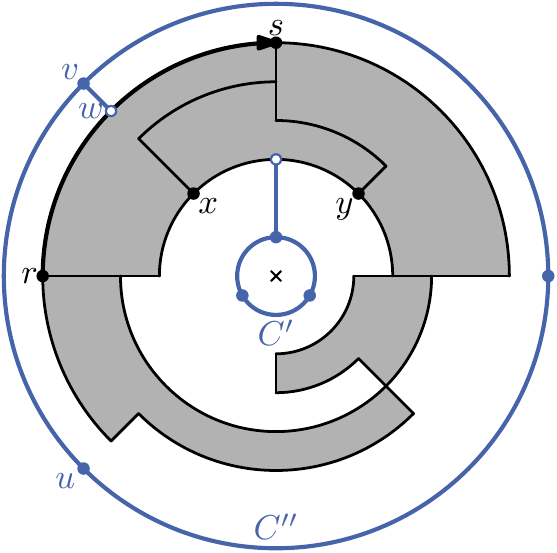}
 \caption{The outer and the central face are rectangulated by adding cycles of 
 length~3.
 The cycle~$C'$ is connected to an arbitrary edge~$xy$ that has label~0 and 
 $C''$ is connected to a new vertex on the old reference edge~$rs$.
 The edge~$uv$ is selected as the new reference edge.}
 \label{fig:rect:outer_central_face}
\end{figure}

If the central and outer faces are not rectangles, we insert triangles in both 
faces and suitably connect these to the original graph; see 
Fig.~\ref{fig:rect:outer_central_face}.
For the central face~$g$ we identify an edge $e$ on the simple cycle~$C$ 
bounding $g$ such that $\ell_C(e)=0$. Since $\Gamma$ is valid and $C$ is an 
essential cycle, such an edge must exist.
We then insert a new cycle~$C'$ of length~3 inside $g$ and connect one of its 
vertices to a new vertex on $e$. The new cycle~$C'$ now forms the boundary of 
the central face.Analogously, we insert into the outer face a cycle~$C''$ of 
length $3$ which contains $G$ and is connected to the reference edge.

In the remainder we therefore only need to consider regular faces.
By definition of a rectangle, as long as there is a regular face that
is not a rectangle, there is a left turn on that face. We augment that
face with an additional edge such that the left turn is resolved and
no further left turns are introduced, which guarantees that the
procedure of augmenting the graph terminates.

We further argue that we augment the given representation in such a
way that it remains valid. Conditions~\ref{cond:repr:sum_of_angles}
and \ref{cond:repr:rotation_faces} of
Definition~\ref{def:repr:valid_representation} are easy to preserve if we
choose the targets of the augmentation correctly; the proof is
analogous to the proof by Tamassia~\cite{t-emn-87}. In particular, the next
observation helps us to check for these two conditions.
\begin{observation}\label{obs:rect:augmentation_conditions_1_to_4}
The representation $\Gamma^u_{vw}$ satisfies Conditions~\ref{cond:repr:sum_of_angles} and~\ref{cond:repr:rotation_faces} of Definition~\ref{def:repr:valid_representation} if and only if $\rot(\subpath{f}{u,vw})=2$.
\end{observation}
We further have to check Condition~\ref{cond:repr:labeling}, i.e., we
need to ensure that we do not introduce monotone cycles when
augmenting the graph.

We now describe the approach in more detail. Tamassia~\cite{t-emn-87}
shows that if there is a left turn, there also is a left turn on a
face $f$ such that the next two turns on $f$ are right turns.
We resolve that left turn by augmenting~$G$, which we sketch in the
following. Let $u$ be the vertex at which face $f$ takes that left
turn. Let $t$ be the vertex before $u$ on $f$. We differentiate two
major cases, namely that $tu$ is either vertical or horizontal.

\textbf{Case 1, $tu$ is vertical.}  Let $vw$ be the edge that appears
on $f$ after the two right turns following the left turn at~$u$. We
show that creating a new vertex $z$ on $vw$ and adding the edge $uz$
(cf. to Fig.~\ref{fig:rect:insertion-ver} for an example) always
upholds validity; see~Section~\ref{sec:rect:correctness} for details.

\textbf{Case 2, $tu$ is horizontal.} This case is more intricate,
because we may introduce decreasing or increasing cycles by augmenting
the graph. Fig.~\ref{fig:rect:insertion-hor_dec} shows an example,
where the graph does not include a descending cycle without the edge
$uz$, but inserting $uz$ introduces one. We therefore do not just
consider $\Gamma^{u}_{vw}$, but $\Gamma^{u}_{e}$ for a set of
\emph{candidate edges} $e$: These include all edges $v'w'$ on $f$ such
that $\rot(\subpath{f}{u,v'w'})=2$. Hence,
Condition~\ref{cond:repr:sum_of_angles} and
Condition~\ref{cond:repr:rotation_faces} are satisfied by
Observation~\ref{obs:rect:augmentation_conditions_1_to_4}.  If there
is a candidate $e$ such that $\Gamma^{u}_{e}$ does not contain a
monotone cycle, we are done.

So assume that there is no such candidate and, furthermore, assume
without loss of generality that $tu$ points to the right. Let the set
of candidates be ordered as they appear on $f$. We show that
introducing an edge from $u$ to the first candidate never introduces
an increasing cycle, while introducing an edge from $u$ to the last
candidate never introduces a decreasing cycle. This also implies that
there must be a consecutive pair of candidates $vw$ and $v'w'$ such
that introducing an edge from $u$ to $vw$ creates a decreasing cycle
and introducing an edge from~$u$ to $v'w'$ creates an increasing
cycle.

In Section~\ref{sec:rect:correctness} we show that in that case, one of the edges
$uw$ or $uv'$ can be inserted without introducing a monotone
cycle. Together with
Observation~\ref{obs:rect:augmentation_conditions_1_to_4}, this
concludes the proof that we can always remove the left turn at $u$
from the representation while maintaining its validity.

Altogether, the algorithm consists of first finding a suitable left
turn in the representation, then determining which of the two cases
applies and finally performing the augmentation. In particular, when
augmenting the graph, we need to ensure that we do not introduce
monotone cycles. Checking for monotone cycles can trivially be done by
testing all essential cycles. However, this may require an exponential
number of tests and it is unknown whether this test can be done in
polynomial time.

\begin{figure}[tb]
 \centering
 \includegraphics{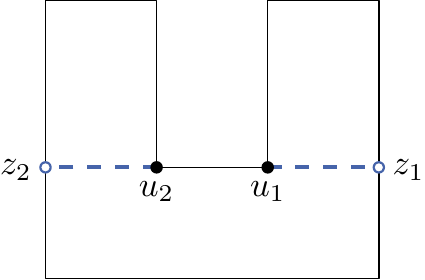}
 \hspace{1ex}
 \captionof{figure}{The face~$f$ is shaped such that the edge that shall be 
 inserted from a left turn must point to the left or to the right.
 The left turn at $u_1$ is followed by two right turns while the left turn at 
 $u_2$ is preceded by two right turns.}
 \label{fig:rect:only_horizontal}
\end{figure}

Therefore, one might wish to avoid the insertion of horizontal edges.
However, faces can be shaped in which horizontal insertions are necessary, 
even if one additionally searches for left turns that are preceded (and not 
only followed) by two right turns.
For example, the U-shaped face in Fig.~\ref{fig:rect:only_horizontal} only 
admits the insertion of new edges that point left or right.

\begin{figure}[bt]
 \centering
 \subfloat[$\Gamma^u_{vw}$]{   
 \includegraphics[scale=1.2]{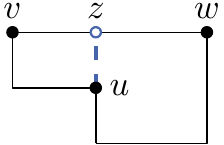}
   \label{fig:rect:insertion-ver}
  }
 \hfill
 \subfloat[$\Gamma^u_{vw}$]{ 
 \includegraphics[scale=1.2]{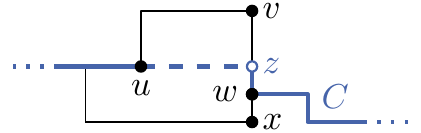}
  \label{fig:rect:insertion-hor_dec}
 }
 \hfill
 \subfloat[$\Gamma^u_{wx}$]{ 
 \includegraphics[scale=1.2]{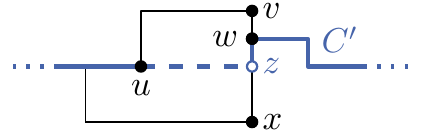}
 \label{fig:rect:insertion-hor_ok}
 }

 \caption{Examples of augmentations.
 \protect\subref{fig:rect:insertion-ver}~Inserting $uz$ is valid, if $uz$ 
 points upwards.
 \protect\subref{fig:rect:insertion-hor_dec}~The representation 
 $\Gamma^u_{vw}$ is not valid since inserting the new edge introduces a 
 decreasing cycle $C$.
 \protect\subref{fig:rect:insertion-hor_ok}~The candidate $wx$ instead gives 
 the valid representation $\Gamma^u_{wx}$. The cycle $C'$, which uses the same 
 edges outside of $f$ as $C$ before, is neither in- nor decreasing. }
 \label{fig:rect:insertion}
\end{figure}

\subsection{Correctness of Rectangulation Algorithm}\label{sec:rect:correctness}
Following the proof structure outlined in the previous section we argue more 
detailedly that a valid augmentation always exists.
We start with Case 1, in which the inserted edges is vertical and show that a 
valid augmentation always exists.

\newcommand{\lemInsertionVerticalEdge}{
Let $vw$ be the first candidate edge after $u$.
If the edge of $f$ entering $u$ points up or down, $\Gamma^u_{vw}$ is a valid ortho-radial representation of the augmented graph $G+uz$.
}

\begin{lemma}\label{lem:rect:insertion_of_vertical_edge}
\lemInsertionVerticalEdge
\end{lemma}

\begin{proof}

  Assume that $\Gamma^u_{vw}$ contains a simple increasing or
  decreasing cycle $C$. As $\Gamma$ is valid, $C$ must contain the new
  edge $uz$ in either direction (i.e., $uz$ or $zu$).  Let $f'$ be the
  new rectangular face of $G+uz$ containing $u$, $v$ and $z$, and
  consider the subgraph $H=C+f'$ of $G+uz$.  According to
  Lemma~\ref{lem:rect:existence_C'} there exists a simple essential
  cycle $C'$ that does not contain $uz$. Moreover, $C'$ can be
  decomposed into paths $P$ and $Q$ such that $P$ lies on $f'$ and $Q$
  is a part of $C$.

The goal is to show that $C'$ is increasing or decreasing. We present
a proof only for the case of increasing $C$. The proof for decreasing
cycles can be obtained by flipping all inequalities.

For each edge $e$ on $Q$ the labels $\ell_C(e)$ and $\ell_{C'}(e)$ are
equal by Lemma~\ref{lem:repr:equal_labels_at_intersection}, and hence $\ell_{C'}(e)\leq 0$.  For
an edge $e\in P$, there are two possible cases: $e$ either lies on the
side of $f'$ parallel to $uz$ or on one of the two other sides.  In
the first case, the label of $e$ is equal to the label $\ell_C(uz)$
($\ell_C(zu)$ if $C$ contains $zu$ instead of $uz$). In particular the
label is negative.

In the second case, we first note that $\ell_{C'}(e)$ is even, since
$e$ points left or right.  Assume that $\ell_{C'}(e)$ was positive and
therefore at least $2$. Then, let $e'$ be the first edge on $C'$ after
$e$ that points to a different direction. Such an edge exists, since
otherwise $C'$ would be an essential cycle whose edges all point to
the right but they are not labeled with 0.
This edge $e'$ lies on $Q$ or is parallel to $uz$. Hence, the argument above implies that $\ell_{C'}(e')\leq 0$. However, $\ell_{C'}(e')$ differs from $\ell_{C'}(e)$ by at most $1$, which requires $\ell_{C'}(e')\geq 1$.
Therefore, $\ell_{C'}(e)$ cannot be positive.

We conclude that all edges of $C'$ have a non-positive label. If all
labels were $0$, $C$ would not be an increasing cycle by
Lemma~\ref{lem:rect:two_cycles_horizontal}.  Thus, there exists an
edge on $C'$ with a negative label and $C'$ is an increasing cycle in
$\Gamma$. But as $\Gamma$ is valid, such a cycle does not exist, and
therefore $C$ does not exist either. Hence, $\Gamma^u_{vw}$ is valid.
\end{proof}

In Case 2 the inserted edge is horizontal. If there is a candidate $e$
such that $\Gamma^{u}_e$ is valid, we choose this augmentation and the
left turn at $u$ is removed successfully.  Otherwise, we make use of
the following structural properties. As mentioned before we assume in
the following that $uz$ points to the right; the other case can be
treated analogously.

\begin{lemma}\label{lem:rect:first_candidate}
Let $vw$ be the first candidate after $u$. No increasing cycle exists in 
$\Gamma^u_{vw}$.
\end{lemma}

\begin{proof}
Let $f'$ be the new rectangular face of $\Gamma^u_{vw}$ containing $u$, $v$ 
and $z$, and assume that there is an increasing cycle $C$ in $\Gamma^u_{vw}$. 
This cycle must use either $uz$ or $zu$.
Similar to the proof of Lemma~\ref{lem:rect:insertion_of_vertical_edge}, we 
find an increasing cycle~$C'$ in $\Gamma$, contradicting the validity of 
$\Gamma$.

To this end, we construct the graph $H$ as the union of $C$ and $f'$.
By Lemma~\ref{lem:rect:existence_C'}, there exists a simple essential cycle 
$C'$ without $uz$ and $zu$ that can be decomposed into a path $P$ on $f'$ and 
a path $Q\subseteq C\setminus f'$ such that all edges of $Q$ have non-positive 
labels.
We show in the following that the edges of $P$ also have non-positive labels.

\begin{figure}[bt]
 \centering
 \includegraphics{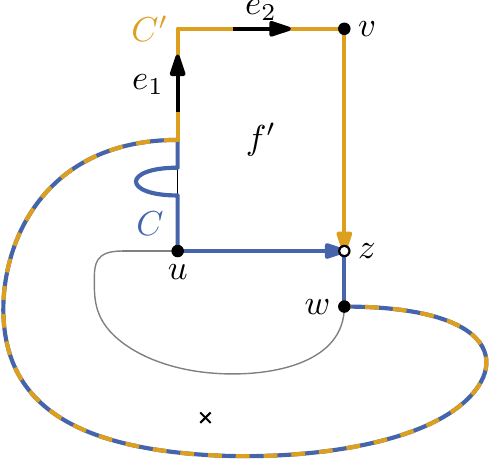}
 \caption{The increasing cycle $C$ contains $uz$. There are three 
 possibilities for edges on $C'$ that lie not on $C$: They lie on the left 
 side of $f'$ (like $e_1$), on the top (like $e_2$), or on the right side 
 formed by only the edge $vz$.}
 \label{fig:rect:first_candidate}
\end{figure}

If $C$ contains $uz$, there are three possibilities for an edge $e$ of $P$, 
which are illustrated in Fig.~\ref{fig:rect:first_candidate}: The edge~$e$ 
lies on the left side of $f'$ and points up, $e$ is parallel to $uz$, or 
$e=vz$.
In the first case $\ell_{C'}(e)=\ell_{C}(uz)-1<0$ and in the second case 
$\ell_{C'}(e)=\ell_C(uz)\leq 0$. If $e=vz$, $C$ cannot contain $zv$ and 
therefore $zw\in C$.
Then, $\ell_{C'}(e)=\ell_C(zw)< 0$. In all three cases the label of $e$ is at 
most 0.

If $C$ contains $zu$, the label of $zu$ has to leave a remainder of 2 when it 
is divided by 4 since $zu$ points to the right.
As the label is also at most $0$, we conclude $\ell_C(zu)\leq -2$. The edges 
of $P$ lie either on the left, top or right of $f'$. Therefore, the label of 
any edge $e$ on $P$ differs by at most 1 from $\ell_{C}(zu)$, and thus we get 
$\ell_{C'}(e)\leq 0$.

Summarizing the results above, we see that all edges on $C'$ are labeled with 
non-positive numbers. The case that all labels of $C'$ are equal to 0 can be 
excluded, since $C$ would not be an increasing cycles by 
Lemma~\ref{lem:rect:two_cycles_horizontal}.
Hence, $C'$ is an increasing cycle, which was already present in $\Gamma$, 
contradicting the validity of $\Gamma$.
Thus, no such increasing cycle~$C$ in $\Gamma^u_{vw}$ exists.
\end{proof}

In order to obtain a similar result for the last candidate, we give a 
sufficient condition for the existence of a candidate after an edge.

\begin{lemma}\label{lem:rect:candidate_after_small_rotation}
Let $f$ be a regular face of $G$ and $u$ be any vertex on $f$.
If $e$ is an edge on $f$ such that $\rot(\subpath{f}{u,e})\leq 2$, there is a 
candidate on $\subpath{f}{e, u}$.
\end{lemma}

\begin{proof}
  For each edge $e'$ on $f$ we determine
  $r(e'):=\rot(\subpath{f}{u,e'})$.  By assumption, it is $r(e)\leq
  2$. For the last edge $e_\text{last}$ on $\subpath{f}{e,
    u}$ it is $r(e_\text{last})=\rot(f)-\rot(u)\geq 3$.  We use
  that $f$ is a regular face (i.e., $\rot(f)=4$) and the rotation
  $\rot(u)$ at the vertex $u$ in $f$ is at most 1.

When going from an edge~$e_1$ to the next edge~$e_2$, the value assigned to 
these edges increases by at most 1, i.e., $r(e_2)\leq r(e_1)+1$. Therefore, 
there exists an edge~$e'$ between $e$ and $e_\text{last}$, i.e., on 
$\subpath{f}{e, u}$ such that $r(e')=2$.
\end{proof}

We now show that augmenting to the last candidate does not create any 
decreasing cycles.

\begin{lemma}\label{lem:rect:last_candidate}
Let $vw$ be the last candidate before $u$. No decreasing cycle exists in 
$\Gamma^u_{vw}$.
\end{lemma}

\begin{proof}
\begin{figure}[bt]
 \centering
 \includegraphics{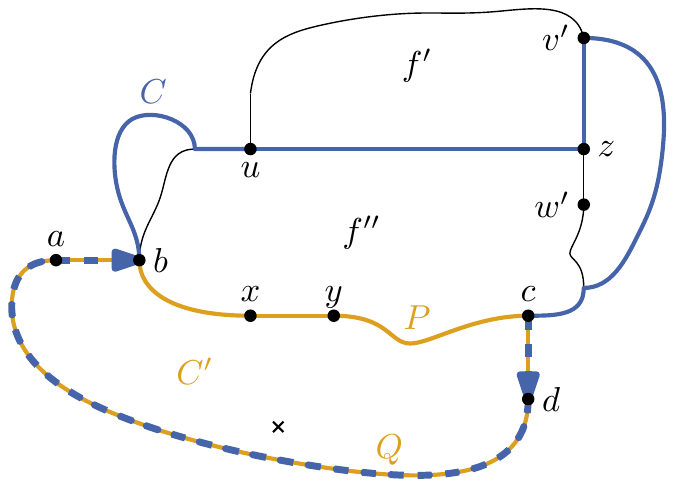}
 \caption{The situation in the proof of Lemma~\ref{lem:rect:last_candidate}. 
 The cycle $C$ is decreasing and it is assumed that $\ell_{C'}(xy)<0$.}
 \label{fig:rect:last_candidate}
\end{figure}

The face $f$ is split in two parts when $uz$ is inserted. Let $f'$ be the face 
containing $v$ and $f''$ the one containing $w$. Assume that there is a simple 
decreasing cycle~$C$ in $\Gamma^u_{vw}$.
Then, either $uz$ or $zu$ lies on $C$. We use a similar strategy as in the 
proof of Lemma~\ref{lem:rect:insertion_of_vertical_edge} to find a decreasing 
cycle in $\Gamma$, which contradicts that $\Gamma$ is valid.

Consider the graph $H=f''+C$ composed of the decreasing cycle $C$ and $f''$. 
Lemma~\ref{lem:rect:existence_C'} shows that there exists an essential 
cycle~$C'$ in $H$ that can be decomposed into a path $P$ on $f''$ and $Q=C\cap 
C'$ (see Fig.~\ref{fig:rect:last_candidate}).
We show in the following that $C'$ is a decreasing cycle. But all edges of 
$C'$ are already present in $\Gamma$, contradicting the assumption that 
$\Gamma$ is valid.
For all edges $e\in E(Q)$, we have $\ell_C(e)=\ell_{C'}(e)\geq 0$ by 
Corollary~\ref{cor:rect:existence_C'_intersection}.

To show that edges on $P$ also have non-negative labels, we assume this was 
not the case, i.e., there is an edge $xy\in P$ such that $\ell_{C'}(xy)<0$.
We present a detailed argument for the case that $C$ uses $uz$ (and not $zu$). 
Then, $P$ is directed such that $f''$ lies to the left of $P$. At the end of 
the proof, we briefly outline how the argument can be adapted if $C$ uses $zu$.

Our goal is to show that there must be a candidate on $f$ after $y$ and in 
particular after the last candidate $vw$---a contradiction.
To this end, we first assume that $\rot(\subpath{f}{u,yx} < 3)$, which we 
prove later on.
By Lemma~\ref{lem:rect:candidate_after_small_rotation} there then is a 
candidate on $\subpath{f}{yx, u}$.
In particular, this candidate comes after $vw$, contradicting the assumption 
that $vw$ is the last candidate.

Moreover, not all labels of edges on $C'$ can be 0, since $C$ would not be 
decreasing (cf.~Lemma~\ref{lem:rect:two_cycles_horizontal}).
Thus, $C'$ is a decreasing cycle that consists solely of edges of~$G$. But as 
$\Gamma$ is valid such a cycle cannot exist, showing that $\Gamma^u_{vw}$ 
contains no decreasing cycle either.

It remains to show the upper bound $\rot(\subpath{f}{u,yx}) < 3$. Let $ab$ and 
$cd$ be the last and the first edge of $Q$, respectively.
In order to simplify the descriptions of the paths we use below, we assume 
without loss of generality that $ab$ and $cd$ have no common endpoints.
This property does not hold, only if $Q$ has at most two edges, in which case 
we may subdivide an edge of $Q$ to lengthen $Q$.

Applying Lemma~\ref{lem:repr:rot_paths_cycle} to $\subpath{C}{a,d}$ and 
$\subpath{C'}{a,d}$, we get
\begin{equation}
\rot(\subpath{C}{a,d}) = \rot(\subpath{C'}{a,d}).
\end{equation}
Splitting the paths at $uz$ and $xy$, respectively, the total rotation does 
not change (cf.~Observation~\ref{obs:rot_splitting_path}).
\begin{align}
\rot(\subpath{C}{a,d})  &=\rot(\subpath{C}{a,uz}) + \rot(\subpath{C}{uz,d}) \\
\rot(\subpath{C'}{a,d}) &=\rot(\subpath{C'}{a,xy}) + \rot(\subpath{C'}{xy,d})
\end{align}
Combining the previous equations gives
\begin{equation}\label{eqn:rect:last_candidate:first_sum}
\rot(\subpath{C}{a,uz}) + \rot(\subpath{C}{uz,d}) - \rot(\subpath{C'}{a,xy}) - 
\rot(\subpath{C'}{xy,d}) = 0.
\end{equation}
As $\subpath{C'}{xy,d}=\subpath{\reverse{f''}}{xy,c}\join cd$, we get
\begin{equation}\label{eqn:rect:last_candidate:C'xy_d}
\rot(\subpath{C'}{xy,d})=\rot(\subpath{\reverse{f''}}{xy,c}+cd) = 
-\rot\left(dc\join\subpath{f''}{c,yx}\right).
\end{equation}
The last equality follows from the fact that the rotation of the reverse of a 
path is the negative rotation of the path 
(cf.~Observation~\ref{obs:rot_reverse}).
Applying Lemma~\ref{lem:repr:rotation_paths_to_cycle} to $\subpath{C}{uz,d}$ 
and $\subpath{f''}{uz,c}\join cd$, we get
\begin{equation}\label{eqn:rect:last_candidate:Cuz_d}
\rot(\subpath{C}{uz,d}) = \rot(\subpath{f''}{uz, c}+cd).
\end{equation}
Substituting Equations~\ref{eqn:rect:last_candidate:C'xy_d} and 
\ref{eqn:rect:last_candidate:Cuz_d} into 
Equation~\ref{eqn:rect:last_candidate:first_sum} yields
\begin{equation}\label{eqn:rect:last_candidate:second_sum}
\rot(\subpath{C}{a,uz}) + \rot(\subpath{f''}{uz,c}\join cd) - 
\rot(\subpath{C'}{a,xy}) + \rot(dc\join\subpath{f''}{c,yx}) = 0.
\end{equation}
Note that if one joins $\subpath{f''}{uz,c}\join cd$ and 
$dc\join\subpath{f''}{c,yx}$ together, the resulting path is almost 
$\subpath{f''}{uz, yx}$ except for the detour via $d$.
Observation~\ref{obs:pre:rot_path_detour} therefore implies
\begin{equation}
\rot\left(\subpath{f''}{uz,c}\join cd\right) + 
\rot\left(dc\join\subpath{f''}{c,yx}\right) = \rot(\subpath{f''}{u, yx})-2.
\end{equation}
Substituting this equality in 
Equation~\ref{eqn:rect:last_candidate:second_sum} and rearranging results in
\begin{equation}
\rot(\subpath{f''}{u, yx}) = 2 + \rot(\subpath{C'}{a,xy}) - 
\rot(\subpath{C}{a,uz}).
\end{equation}
By Observation~\ref{obs:repr:label_difference} the rotations on an essential 
cycle can be expressed as the difference of labels.
Here, we have $\rot(\subpath{C'}{a,xy})=\ell_{C'}(xy)-\ell_{C'}(ab)$ and 
$\rot(\subpath{C}{a,uz})=\ell_C(uz)-\ell_C(ab)$.
Additionally $ab$ lies on $Q$, and therefore it is $\ell_C(ab)=\ell_{C'}(ab)$. 
Hence,
\begin{align}
\rot(\subpath{f''}{u, yx}) &= 2 + 
\ell_{C'}(xy)-\ell_{C'}(ab)-\ell_C(uz)+\ell_C(ab) \nonumber\\
 &= 2+\ell_{C'}(xy)-\ell_C(uz) \nonumber\\
 &< 2.
\end{align}
By construction of $\Gamma^u_{vw}$ the rotations of $\subpath{f}{u, yx}$ and 
$\subpath{f''}{u, yx}$ differ by exactly 1, i.e.,
\begin{equation}
\rot\subpath{f}{u,yx}=\rot\subpath{f''}{u,yx}+1<3.
\end{equation}

We assumed that $C$ uses $uz$ in that direction. If $zu$ lies on $C$, a 
similar argument shows that $C'$ would also be a decreasing cycle, 
contradicting that assumption that $\Gamma$ is valid.
\end{proof}

By assumption $\Gamma^{u}_e$ contains a monotone cycle for all
candidate edges $e$. By Lemma~\ref{lem:rect:first_candidate} the
augmentation for the first candidate contains a decreasing cycle and
by Lemma~\ref{lem:rect:last_candidate} the augmentation for the last
candidate has an increasing cycle. Hence, there are consecutive
candidates $vw$ and $v'w'$ such that $\Gamma^u_{vw}$ has a decreasing
cycle and $\Gamma^u_{v'w'}$ an increasing cycle.

\begin{figure}[bt]
 \centering
 \subfloat[]{ \label{fig:rect:horizontal_between-neighbors}\includegraphics{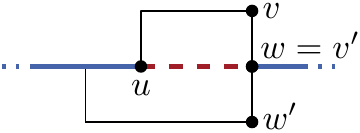} }
 \hfill
 \subfloat[]{ \label{fig:rect:horizontal_between-w}\includegraphics{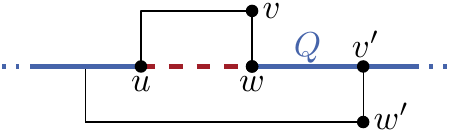} }
 \hfill
 \subfloat[]{ \label{fig:rect:horizontal_between-v}\includegraphics{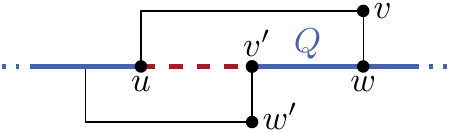} }%
  \caption{Three possibilities how the path between $w$ and $v'$ can look 
  like: \protect\subref{fig:rect:horizontal_between-neighbors}~$w=v'$, 
  \protect\subref{fig:rect:horizontal_between-w}~all edges point right, and 
  \protect\subref{fig:rect:horizontal_between-v} all edges point left. In the 
  first two cases the edge $uw$ is inserted and in (c) $uv'$ is added.}
 \label{fig:rect:horizontal_between}
\end{figure}

\begin{lemma}\label{lem:rect:between_decreasing_and_increasing}
Let $Q$ be the path on $f$ between the candidates $vw$ and $v'w'$.
Then, there is a path~$P$ in $G$ containing $w$, $v'$ and $u$ such that the 
edges point to the right.
More precisely, $P$ starts at $w$ or $v'$ and ends at $u$, and either $Q$ or 
$\reverse{Q}$ forms the first part of $P$. Moreover, the start vertex of $P$ 
has no incident edge to its left.
\end{lemma}

\begin{proof}
\begin{figure}[tb]
\centering
\includegraphics{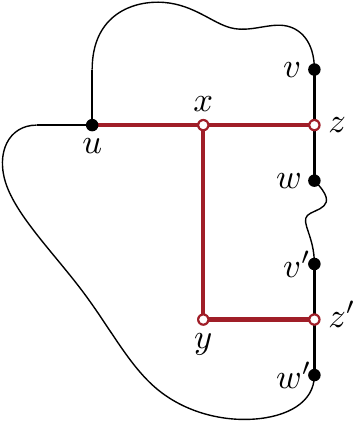}
\caption{The structure that is used to simulate the insertion of both $uz$ and 
$uz'$ at the same time. The edge $uz$ is replaced by the path $uxz$ and $uz'$ 
by $uxyz'$.}
\label{fig:rect:between_structure}
\end{figure}

Let $z$ be the new vertex inserted in $\Gamma^u_{vw}$ and $z'$ the one in 
$\Gamma^u_{v'w'}$.
Since both $uz$ and $uz'$ point to the right, there is no augmentation of 
$\Gamma$ containing both edges.
We need to compare $\Gamma^u_{vw}$ and $\Gamma^u_{v'w'}$ though.
Therefore, we use the following construction, which models all important 
aspects of both representations:
Starting from $\Gamma$ we insert new vertices $z$ on $vw$ and $z'$ on $v'w'$. 
We connect $u$ and $z$ by a path of length 2 that points to the right and 
denote its internal vertex by $x$.
Furthermore, a path of length~2 from $x$ via a new vertex $y$ to $z'$ is added.
The edge $xy$ points down and $yz'$ to the right. In the resulting 
ortho-radial representation $\tilde{\Gamma}$ the edge $uz$ is modeled by the 
path $uxz$ and $uz'$ by $uxyz'$.
The resulting construction is depicted in 
Fig.~\ref{fig:rect:between_structure}.

Take any simple decreasing cycle in $\Gamma^u_{vw}$. As $\Gamma$ is valid, 
this cycle must contain either $uz$ or $zu$.
We obtain a cycle~$C$ in $\tilde\Gamma$ by replacing $uz$ with $uxz$ (or $zu$ 
with $zxu$).
Note that $ux$ and $xz$ have the same label as $uz$, and the labels of all 
other edges on the cycles stay the same. Therefore, $C$ is a decreasing cycle.

Similarly, there exists a simple increasing cycle in $\Gamma^u_{v'w'}$, which 
contains $uz'$ or $z'u$.
Replacing $uz'$ with $uxyz'$ (or $z'u$ with $z'yxu$) we get a cycle $C'$ in 
$\tilde\Gamma$.
Note that $C'$ might not be an increasing cycle as $\ell_C(xy)$ might be 
positive. But the labels of all other edges are at most 0 and there exists an 
edge with a negative label.
In other words, outside of $f$ the cycle $C'$ behaves exactly like an 
increasing cycle.

For now we assume that the original cycles use $uz$ and $uz'$ in these 
directions. At the end of the proof, we shall see that this is in fact the 
only possibility.
The proof is structured as follows: First, we show 
$\ell_C(ux)=\ell_{C'}(ux)=0$. This also determines the labels of $C$ and $C'$ 
in the interior of $f$.
In a second step, we find that at least one of the vertices $w$ and $v'$ has 
no incident edge to the left and the other vertex lies on both $C$ and $C'$.
Note that $w=v'$ is possible. In this case, $w=v'$ has all the properties 
mentioned.
Moreover, we prove that the path~$Q$ on $f$ between $w$ and $v'$ (of length~0 
if $w=v'$) is straight and can be used as the first part of the desired path 
$P$.
From this information we can infer that there are three possibilities as shown 
in Fig.~\ref{fig:rect:horizontal_between}: Either $w=v'$, all edges on $f$ 
between $w$ and $v'$ point right, or they all point left.
Finally, we show that outside of $f$, the cycles $C$ and $C'$ are equal. In 
particular, all their edges point to the right and we can use them as the 
second part of the desired path~$P$.

\begin{figure}[bt]
 \centering
 \includegraphics{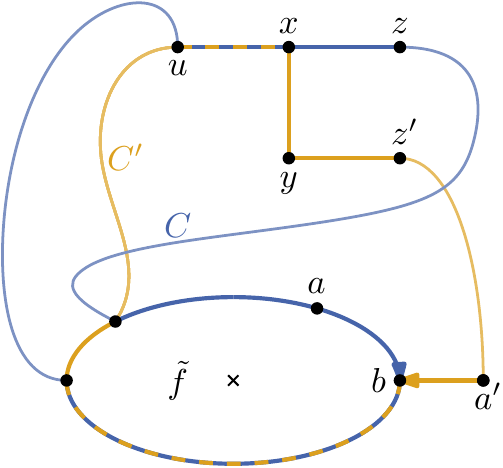}
 \caption{If $ux$ did not lie on the central face~$\tilde f$, there would be 
 edges $ab$ on $C$ and $a'b$ on $C'$ such that $ab$ would lie on $\tilde f$ 
 but $a'b$ does not.
 But as $\ell_C(ab)\geq 0$ and $\ell_{C'}(a'b)\leq 0$, the edge $a'b$ would 
 have to lie inside $\tilde f$. Hence, this situation does not occur.}
 \label{fig:rect:between_impossible_central_face}
\end{figure}

To show $\ell_C(ux)=\ell_{C'}(ux)$, we consider the graph $H=C+C'$ formed by 
the two cycles $C$ and $C'$ and denote its central face by~$\tilde f$.
By Lemma~\ref{lem:repr:equal_labels_at_intersection} it suffices to show that 
$ux$ lies on $\tilde f$. Assume for the sake of contradiction that this was 
not the case.
Then, $xy$, $xz$ and $yz'$ do not lie on $\tilde f$ either.
Hence, $\tilde f$ is formed completely by edges in $E(\subpath{C}{z,u})\cup 
E(\subpath{C'}{z', u})$. As $C$ and $C'$ were constructed by subdividing edges 
of simple cycles, they are simples themselves.
Therefore, the edges of $\tilde f$ do not all belong to the same cycle.
Therefore, there is an edge $ab$ on $\tilde f$ such that $b$ lies on both $C$ 
and $C'$ but $ab$ only belongs to $C$ and not to $C'$ 
(cf.~Fig.~\ref{fig:rect:between_impossible_central_face}).
Since $C$ is a decreasing cycle, it is $\ell_C(ab)\geq 0$.

Moreover, let $a'$ be the vertex before $b$ on $C'$.
Since $C'$ is almost an increasing cycle, we get $\ell_{C'}(a'b)\leq 0$ unless 
$a'b=xy$. But this is impossible since $y$ does not lie on $\tilde f$.
Lemma~\ref{lem:repr:illegal_intersection} therefore implies that $a'b$ lies in 
the interior of $C$.
But then $ab$ would not lie on $\tilde{f}$, contradicting the choice of $ab$. 
Thus, $ux$ is part of $\tilde f$.

Hence, Lemma~\ref{lem:repr:equal_labels_at_intersection} applies to $ux$ and 
we obtain
\[
0\leq \ell_C(ux) = \ell_{C'}(ux)\leq 0.
\]
Thus, $\ell_C(ux) = \ell_{C'}(ux)=0$. Furthermore, we get 
$\ell_C(xz)=\ell_{C'}(yz')=0$.
Therefore, $C$ contains $zw$, because otherwise $zv$ would be labeled with 
$-1$. Similarly, we see that $z'v'$ lies on $C'$.

As a next step we prove that one of $v'$ and $w$ has no incident edge to its 
left and the other vertex lies on both $C$ and $C'$.
This is the case if $v'=w$, since any subgraph to the left of this vertex 
would contain a candidate. Remember that we treat degree-1 vertices as two 
left turns with an edge in between and this edge can be a candidate.
Therefore, even in the extreme case, where the subgraph to the left of $v'=w$ 
is just one path that points left, we find a candidate---namely the leftmost 
endpoint of the path.

\begin{figure}
 \centering
 \subfloat[The path $\subpath{C}{w,t}$ makes a right turn at $v'$. But then 
 $\subpath{C}{z',u}\join uz'$ would be a decreasing cycle in 
 $\Gamma^u_{v'w'}$.]{ \label{fig:rect:between_left_turn_w-turn} 
 \includegraphics[width=0.45\textwidth]{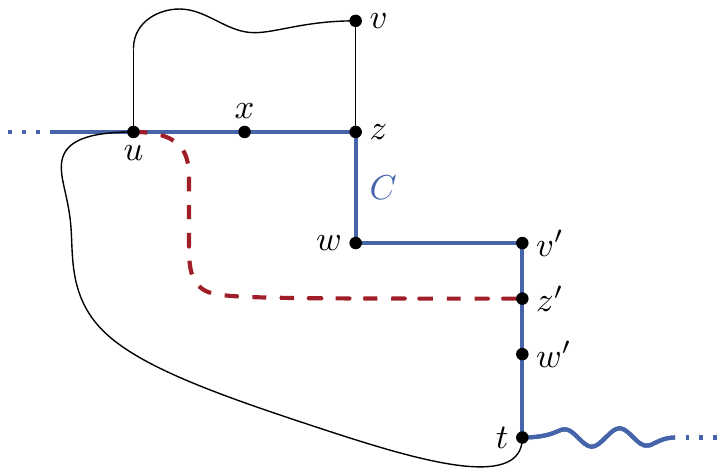}
  }
 \hfill
 \subfloat[The path $\subpath{C}{w,t}$ is straight, which is the only 
 possibility. In this case it is $t=v'$.]{ 
 \label{fig:rect:between_left_turn_w-straight} 
 \includegraphics[width=0.45\textwidth]{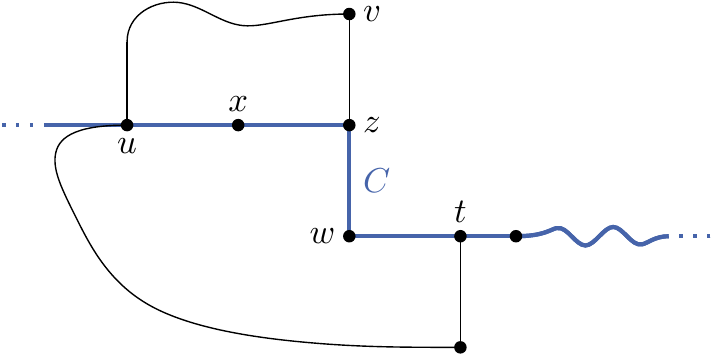}
  }
 \caption{The face $f$ makes a left turn at $w$.}
 \label{fig:rect:between_left_turn_w}
\end{figure}

If $w\neq v'$, $f$ makes either a left or a right turn at $w$.
If $f$ makes a left turn at $w$, $C$ makes a left turn there as well.
Let $t$ be the vertex at which $C$ leaves $f$ after $w$. In other words, $t$ 
is the last vertex of $C$ such that $\subpath{C}{w,t}$ lies on $f$. All edges 
of $f$ starting at $t$ lie to the right of $C$ (i.e., in the interior of $C$).
This situation is illustrated in Fig.~\ref{fig:rect:between_left_turn_w}.

We show that $\subpath{C}{w,t}$ contains $v'$ and that all edges of 
$\subpath{C}{w,v'}$ point to the right.
If $\subpath{C}{w,t}$ makes turns, the first turn cannot be a left turn, since 
the edge following the turn would lie on $C$ and  be labeled with $-1$.
If the first turn is a right turn, the edge following the turn is the 
candidate $v'w'$. But then $\Gamma_u^{v'w'}$ would contain the decreasing 
cycle $\subpath{C}{z',u}\join uz'$, which contradicts the choice of $v'w'$ 
(cf.~Fig.~\ref{fig:rect:between_left_turn_w-turn}).
Hence, it remains the case in which $\subpath{C}{w,t}$ is straight, which is 
illustrated in Fig.~\ref{fig:rect:between_left_turn_w-straight}.
The edge following $t$ on $C$ must point right. Therefore, $f$ makes a right 
turn at $t$ implying $t=v'$.
In all possible cases, $v'$ lies on both $C$ and $C'$ and the path 
$Q:=\subpath{C}{w,v'}$ is straight and points to the right.

If $f$ makes a right turn at $w$, we consider the edge $wa$ following this 
turn.
\[
\rot\subpath{f}{u,wa}=\rot\subpath{f}{u,vw} + \rot(vwa) = 2+1=3
\]
Since $v'w'$ is a candidate, it is $\rot(\subpath{f}{u, v'w'}) = 2$.
When walking along $f$ the rotation changes by at most~1 per step, since 
degree-1 vertices are treated as two steps.
Hence, the rotations of all the edges between $vw$ and $v'w'$ must be at least 
$3$. Otherwise, there would be another candidate in between. Therefore, $f$ 
makes a left turn at~$v'$.

\begin{figure}[bt]
 \centering
 \includegraphics{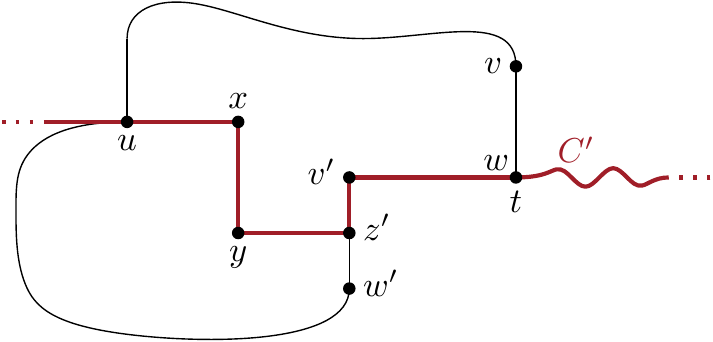}
 \caption{If $f$ makes a right turn at $w$, it makes a left turn at $v'$. The 
 longest common path $\subpath{C'}{v',t}$ on $C'$ and $\reverse{f}$ contains 
 $w$ and all its edges point to the right.}
 \label{fig:rect:between_right_turn_w}
\end{figure}

As $C'$ enters $v'$ on the edge $z'v'$ from below, $C'$ must leave $v'$ to the 
right. Thus, there is a part of $C'$ starting at $v'$ that lies on 
$\reverse{f}$. Let $t$ be the last vertex on $C'$ such that 
$\subpath{C}{v',t}=\subpath{\reverse{f}}{v', t}$.
Fig.~\ref{fig:rect:between_right_turn_w} illustrates the situation.

Similar to the case above, we analyze the first turn of $\subpath{C'}{v',t}$, 
if it exists, and show that $w$ lies on this path. Moreover, we see that all 
edges of this path point to the right.
Because $C'$ is an increasing cycle and the first edge of $\subpath{C'}{v',t}$ 
is labeled with $0$, the first turn cannot be a right turn.
If it is a left turn, the left turn must occur at $w$. Note that in this case 
$\Gamma^u_{vw}$ contains an increasing cycle. We shall see that this situation 
actually cannot occur, but we cannot exclude this possibility yet.
If $\subpath{C'}{v',t}$ makes no turns, the edge on $C'$ after $t$ also points 
to the right and the edge on $f$ incident to $t$ is the candidate $vw$.
Hence, the path $Q:=\subpath{C'}{v',w}$ points completely to the right.

In all cases we found a (possibly empty) path~$Q$ from $v'$ to $w$ or vice 
versa, whose edges point to the right. Moreover, the endpoint~$t$ of this path 
lies on both $C$ and $C'$.
The path $Q$ is the initial part of the desired path~$P$.
To construct the remaining part of $P$, we prove that 
$\subpath{C}{t,u}=\subpath{C'}{t,u}$. Moreover, we show that all edges on this 
path point to the right.

Assume for the sake of contradiction that 
$\subpath{C}{t,u}\neq\subpath{C'}{t,u}$. Then, let $ab$ be the last edge of 
$\subpath{C}{t,u}$ that does not lie on $C'$. Let $a'b$ be the edge of $C'$ 
entering $b$.
By Lemma~\ref{lem:repr:illegal_intersection} the edge $a'b$ lies in the 
interior of $C$. But then consider the last common vertex $c$ of $C$ and $C'$ 
before $b$. We denote the edges of $C$ and $C'$ starting at $c$ by $cd$ and 
$cd'$, respectively. As $ab$ lies in the interior of $C$, so does $cd'$.
But according to Lemma~\ref{lem:repr:illegal_intersection} the edge $ab'$ lies 
in the exterior of $C$, a contradiction.
Thus, all edges of $\subpath{C}{t, u}$ lie on $\subpath{C'}{t, u}$ as well. 
Since these are both paths from $t$ to $u$, this means that they are equal. To 
shorten the notation, we refer to $\subpath{C}{t,u}$ as $R$.

Since $\ell_C(ux)=\ell_{C'}(ux)$, it is $\ell_C(e)=\ell_{C'}(e)$ for all edges 
$e$ on $R$. Moreover, as $C$ is a decreasing and $C'$ an (almost) increasing 
cycle, these labels must be $0$. In particular, all edges on $R$ point to the 
right.

Thus, $P=Q\join R$ is a path containing both $v_{i}$ and $w$ ending at $u$, 
such that all edges of $P$ point to the right, which concludes the proof for 
the case when both $C$ and $C'$ use the edge $ux$ in this direction.

It remains to show that neither $C$ nor $C'$ can contain $xu$. By an argument 
similar to the one above, we know that $xu$ lies on the boundary of the 
central face~$\tilde f$ of $H=C+C'$.
Since $C$ and $C'$ contain $\tilde f$ in their interiors, any edge $e$ on one 
of them incident to $\tilde f$ is directed such that $\tilde f$ lies locally 
to the right of $e$.
If $C$ and $C'$ used $ux$ in different directions, this implies that $\tilde 
f$ lies locally to the right of both $ux$ and $xu$, which is impossible.

Hence, if $xu$ lies on one cycle, it lies on the other one, too. In this case, 
it is $\ell_C(xu)=\ell_{C'}(xu)=0$. But $xu$ points to the left and therefore 
its label must leave a remainder of 2 when divided by 4.
Thus, the assumption that both $C$ and $C'$ contain $ux$ is justified as none 
of the other cases can occur.
\end{proof}

There are three possible ways how the vertices $w$ and $v'$ can be arranged on $P$; see~Fig.~\ref{fig:rect:horizontal_between}.
Either $w=v'$, $w$ comes before $v'$, or $v'$ comes before $w$.
In any case we denote the start vertex of $P$ by $z$. According to Lemma~\ref{lem:rect:between_decreasing_and_increasing}, no edge is incident to the left of~$z$.
Hence, the insertion of the edge~$uz$ such that it points right gives a new ortho-radial representation~$\Gamma'$, which is valid by the following lemma.

\newcommand{\lemHorizontalCycleValid}{
Let $\Gamma'$ be the ortho-radial representation that is obtained from $\Gamma$ by adding the edge~$uz$ pointing to the right as in Lemma~\ref{lem:rect:between_decreasing_and_increasing}. Then, $\Gamma'$ is valid.
}

\begin{lemma}\label{lem:rect:horizontal_cycle_valid}
\lemHorizontalCycleValid
\end{lemma}

\begin{proof}
  By construction, $\Gamma'$ satisfies
  Conditions~\ref{cond:repr:sum_of_angles} and~\ref{cond:repr:rotation_faces}
  of Definition~\ref{def:repr:valid_representation}.  Let $C'=P\join
  uz$ be the new cycle whose edges point right. It is $\ell_{C'}(e)=0$
  for each edge $e$ of $C'$.  No essential cycles without $uz$ or $zu$
  is increasing or decreasing, since they are already present in the
  valid representation~$\Gamma$.  If an essential cycle~$C$ contains
  $uz$ or $zu$ and in particular the vertex~$u$,
  Lemma~\ref{lem:rect:two_cycles_horizontal} states that $C$ is
  neither increasing nor decreasing. Thus, $\Gamma'$ satisfies
  Condition~\ref{cond:repr:labeling} and is therefore valid.
\end{proof}
Putting all results together we see that the rectangulation algorithm
presented in Section~\ref{sec:rect:algorithm} works correctly.  That
is, given a valid ortho-radial representation~$\Gamma$, the algorithm
produces another valid ortho-radial representation~$\Gamma'$ such that
all faces of $\Gamma'$ are rectangles and $\Gamma$ is contained in
$\Gamma'$.  Combining this result with
Corollary~\ref{cor:draw:characterization} we obtain the following theorem.

\begin{theorem}\label{thm:rect:representation_to_drawing}
Let $\Gamma$ be a valid ortho-radial representation of a graph $G$. Then there is a drawing of $G$ representing $\Gamma$.
\end{theorem}

Theorem~\ref{thm:rect:representation_to_drawing} shows one direction of implication of Theorem~\ref{thm:repr:characterization}.
The other direction is proved by the following theorem.

\newcommand{\thmDrawingToRepresentation}{
For any drawing~$\Delta$ of a 4-planar graph~$G$ there is a valid ortho-radial representation of~$G$.
}

\begin{theorem}\label{thm:rect:drawing_to_representation}
\thmDrawingToRepresentation
\end{theorem}

\begin{proof}
\begin{figure}
 \centering
 \subfloat[The result of the geometric rectangulation of an ortho-radial graph drawing.
 The edges of the original graph are drawn in solid black and the edges inserted during the rectangulation in dashed blue.
 The reference edge~$e^\star$ is chosen as an edge on the outermost circle.]{ \label{fig:rect:geometric_rectangulation-complete}
\includegraphics{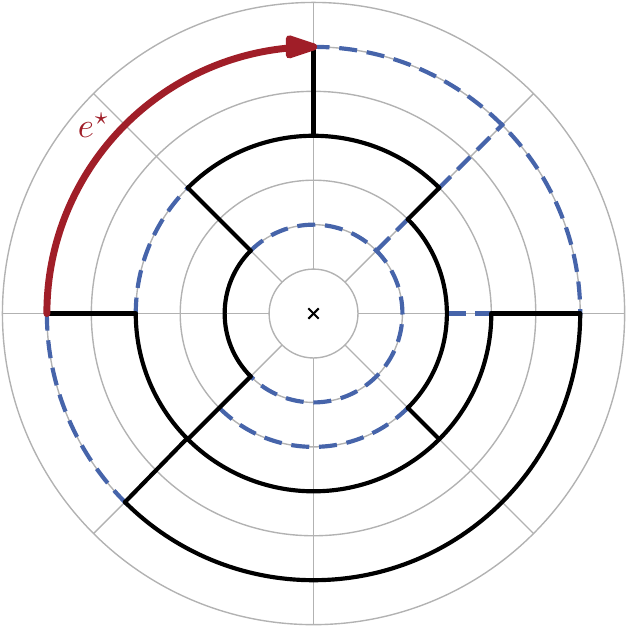} }
 \hspace{1cm}
 \subfloat[One face is rectangulated by casting rays from the left turns until 
 the intersect with an edge.]{ \label{fig:rect:geometric_rectangulation-face} 
 \includegraphics[scale=0.8]{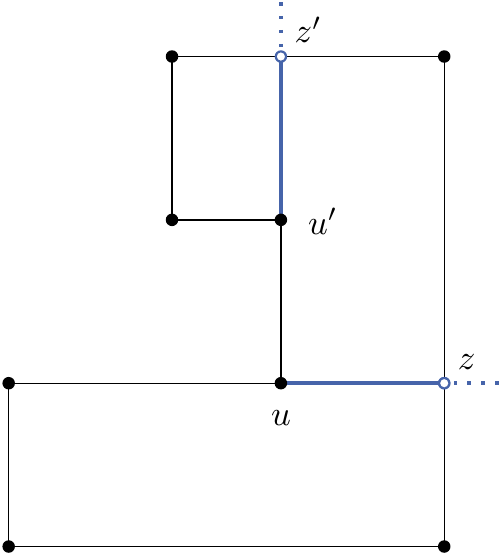}
  }
 \caption{Rectangulations of a graph and a single face}
 \label{fig:rect:geometric_rectangulation}
\end{figure}

We first note that any drawing $\Delta$ fixes an ortho-radial representation up to the choice of the reference edge.
Let $\Gamma$ be such an ortho-radial representation where we pick an 
edge~$e^\star$ on the outer face as the reference edge such that $e^\star$ 
points to the right and lies on the outermost circle that is used by $\Delta$ 
(as in Fig.~\ref{fig:rect:geometric_rectangulation-complete}).
By~\cite{hht-orthoradial-09} the representation~$\Gamma$ satisfies Conditions~\ref{cond:repr:sum_of_angles} and~\ref{cond:repr:rotation_faces} of Definition~\ref{def:repr:valid_representation}.
To prove that $\Gamma$ also satisfies Condition~\ref{cond:repr:labeling}, i.e., $\Gamma$ does not contain any increasing or decreasing cycles, we show how to reduce the general case to the more restricted one, where all faces are rectangles.
By Corollary~\ref{cor:draw:characterization} the existence of a drawing and the validity of the ortho-radial representation are equivalent.

Given the drawing~$\Delta$, we augment it such that all faces are rectangles. This rectangulation is similar to the one described in Section~\ref{sec:rect:algorithm} but works with a drawing and not only with a representation.
We first insert the missing parts of the innermost and outermost circle that are used by $\Delta$ such that the outer and the central face are already rectangles.
For each left turn on a face~$f$ at a vertex $u$, we then cast a ray from $v$ in $f$ in the direction in which the incoming edge of $u$ points (cf.~Fig.~\ref{fig:rect:geometric_rectangulation-face}).
This ray intersects another edge in $\Delta$. Say the first intersection occurs at the point~$p$. Either there already is a vertex~$z$ drawn at $p$ or $p$ lies on an edge. In the latter case, we insert a new vertex, which we call~$z$, at $p$.
We then insert the edge $uz$ in $G$ and update $\Delta$ and $\Gamma$ accordingly.

Repeating this step for all left turns, we obtain a drawing~$\Delta'$ and an ortho-radial representation~$\Gamma'$ of the augmented graph~$G'$ (see~Fig.~\ref{fig:rect:geometric_rectangulation-complete} for an example of $\Delta'$).
As the labelings of essential cycles are unchanged by the addition of edges elsewhere in the graph, any increasing or decreasing cycle in $\Gamma$ would also appear in $\Gamma'$.
But by Corollary~\ref{cor:draw:characterization} $\Gamma'$ is valid, and hence neither $\Gamma$ nor $\Gamma'$ contain increasing or decreasing cycles. Thus, $\Gamma$ satisfies Condition~\ref{cond:repr:labeling} and is valid.
\end{proof}

\section{NP-Hardness Proof for Drawings without Bends} 
\label{sec:bend-minimization}

Garg and Tamassia~\cite{gt-ccurpt-01} showed that it is \NP-complete to decide whether a 4-planar graph admits an orthogonal drawing without any edge bends (\textsc{Othogonal Embeddability}).
In this section, we study the analogous problem for ortho-radial drawings and prove that it is \NP-complete as well.
We say a graph~$G$ admits an \emph{ortho-radial (or orthogonal) embedding}, if there is an embedding of $G$ such that $G$ can be drawn ortho-radially (or orthogonally) without bends.

\begin{definition}[\textsc{Ortho-Radial Embeddability}]
Does a 4-planar graph $G$ admit an ortho-radial embedding?
\end{definition}

To show that \textsc{Ortho-Radial Embeddability} is \NP-hard, we reduce \textsc{Planar Monotone 3-SAT}, which was shown to be \NP-hard by~\cite{l-pftu-82}, to \textsc{Ortho-Radial Embeddability}.
\begin{definition}[\textsc{Planar Monotone 3-SAT}]
Given a Boolean formula~$\Phi$ in conjunctive normal form, such that each clause contains exactly three literals that are either all positive or all negative and a planar representation of its variable-clause-graph in which the variables are placed on one line, the clauses with only positive literals above that line and the clauses with only negative literals below this line. Is $\Phi$ satisfiable?
\end{definition}

To reduce from \textsc{Planar Monotone 3-SAT} to \textsc{Ortho-Radial Embedding} we first construct an equivalent instance~$G$ of \textsc{Orthogonal Embeddability} as described by Bläsius et~al.~\cite{bbr14-kandinsky}.
We then build a structure around $G$ yielding a graph~$G'$ such that in any ortho-radial representation of $G'$ the representation~$\Gamma$ of $G$ does not contain any essential cycles.
In other words, $\Gamma$ is actually an orthogonal representation of $G$. Hence, an ortho-radial embedding of $G'$ can only exist if $G$ admits an orthogonal embedding.
We may assume without loss of generality that $G$ is connected as otherwise, we handle each component separately.

\begin{figure}[bt]
 \centering
 \subfloat[$G$ lies between $C_1$ and $C_2$.]{ \label{fig:embed:embedding-12}\includegraphics{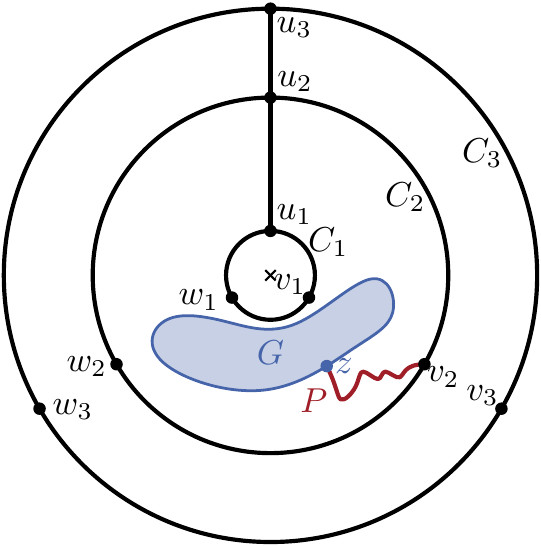} }
 \hspace{10ex}
 \subfloat[$G$ lies between $C_2$ and $C_3$.]{ 
 \label{fig:embed:embedding-23}\includegraphics{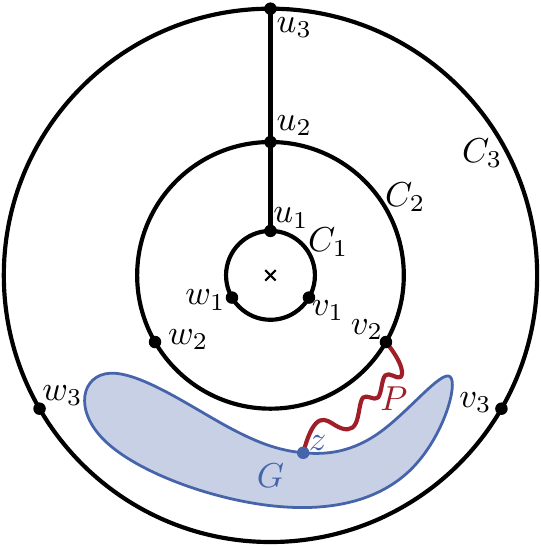}
  }
 \caption{Possible embeddings of $G'$: In both cases $G$ contains no essential cycles. The roles of $C_1$ and $C_3$ can be exchanged.}
 \label{fig:embed:embedding}
\end{figure}

The construction of $G'$ from $G$ is based on the fact that there is only one way to ortho-radially draw a triangle~$C$, i.e., a cycle of length~3, without bends: as an essential cycle on one circle of the grid.
We build a graph $H$ consisting of three triangles called $C_1$, $C_2$ and $C_3$ and denote the vertices on $C_i$ by $u_i$, $v_i$ and $w_i$.
Furthermore, $H$ contains the edges $u_1u_2$ and $u_2u_3$.
In Fig.~\ref{fig:embed:embedding} $H$ is formed by the black edges.

To connect $H$ and $G$, we identify a vertex~$z$ in $G$, called the \emph{port} of $G$, such that either there is an orthogonal embedding of $G$ with $z$ on the outer face and the angle at $z$ in the outer face is at least $180\degree$, or $G$ does not admit an orthogonal embedding.
We shall see later how such a vertex can be identified in polynomial time.
We connect $z$ and $v_2$ by a path~$P$ whose length is equal to the number of edges in $G$ and denote the resulting graph by~$G'$ (cf.~Fig.~\ref{fig:embed:embedding}).

\begin{lemma}\label{lem:embed:equivalence}
The graph $G'$ admits an ortho-radial embedding if and only if $G$ admits an orthogonal embedding with the port~$z$ on the outer face such that the angle at~$z$ in the outer face is at least $180\degree$.
\end{lemma}
\begin{proof}
Let $\Gamma'$ be a valid ortho-radial representation of $G'$ without edge bends. In $\Gamma'$ all three cycles $C_1$, $C_2$ and $C_3$ are essential cycles. Hence, $C_2$ lies between $C_1$ and $C_3$.
Furthermore, $G$ either lies inside the area enclosed by $C_1$, $C_2$ and $u_1u_2$ (as shown in Fig.~\ref{fig:embed:embedding-12}), or in the area enclosed by $C_2$, $C_3$ and $u_2u_3$ (as in Fig.~\ref{fig:embed:embedding-23}).
Therefore, $G$ cannot contain any essential cycles.
Hence, the ortho-radial representation of $G$ can be interpreted as an orthogonal representation, and thus, $G$ admits an orthogonal embedding.
Since $P$ ends at the port~$z$ and intersects $G$ only at $z$, the angle at~$z$ in the outer face of $G$ must be at least $180\degree$.

Let $\Gamma$ be an orthogonal representation of $G$ without bends such that $z$ lies on the outer face and the angle at $z$ in the outer face is at least $180\degree$.
We interpret $\Gamma$ as an ortho-radial representation of $G$ and extend it to a representation of $G'$ as follows:
We embed $H$ such that $C_1$ lies in the interior of $C_2$, which in turn lies in $C_3$. We place $G$ between $C_1$ and $C_2$ as shown in Fig.~\ref{fig:embed:embedding-12}.
Since $\Gamma$ has no bends, the path $P$ connecting $G$ and $H$ needs to make at most as many turns as there are edges incident to the outer face of $G$.
Since the length of $G$ is equal to the number of edges of $G$, we can place $P$ such that all turns occur at its vertices, that is, without edge bends.
As the angle at the port~$z$ has at least $180\degree$, $P$ can lie completely in the exterior of~$G$.
\end{proof}

The construction of $G'$ relies on the fact that we find a vertex~$z$ in $G$ such that placing~$z$ on the outer face of $G$ does not restrict the possible embeddings too much.
In order to show how $z$ can be chosen, we present some details of the reduction from \textsc{Planar Monotone 3-SAT} to \textsc{Orthogonal Embeddability} by Bläsius et~al.~\cite{bbr14-kandinsky}.

\begin{figure}
 \centering
 \subfloat[]{ \label{fig:embed:reduction-3sat}\includegraphics{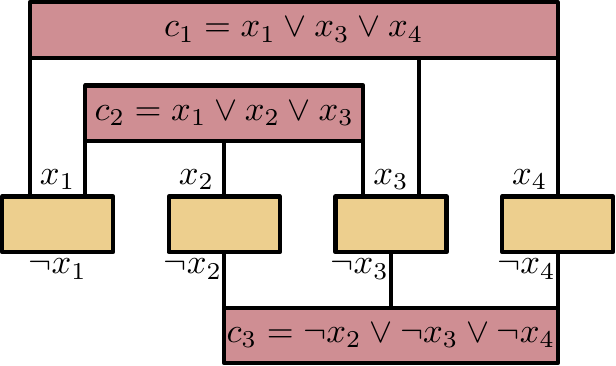} }
 \hfill
 \subfloat[]{ \label{fig:embed:reduction-graph}\includegraphics{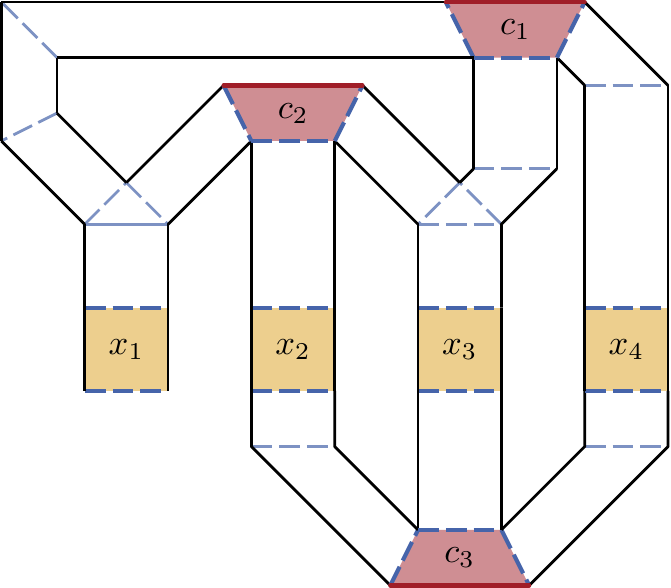} }
 \caption{(a) An instance~$\Phi$ of \textsc{Planar Monotone 3-SAT}. (b) A schematic drawing of the graph $G$ constructed from $\Phi$ as described by Bläsius et~al.~\cite{bbr14-kandinsky}.
 The red-shaded parts correspond to the clauses and the orange-shaded ones to the variables. The dashed blue lines refer to edges with exactly one bend, which encode the truth values.}
 \label{fig:embed:reduction}
\end{figure}

\begin{figure}
 \centering
 \subfloat[]{ \label{fig:embed:gadgets-clause}\includegraphics{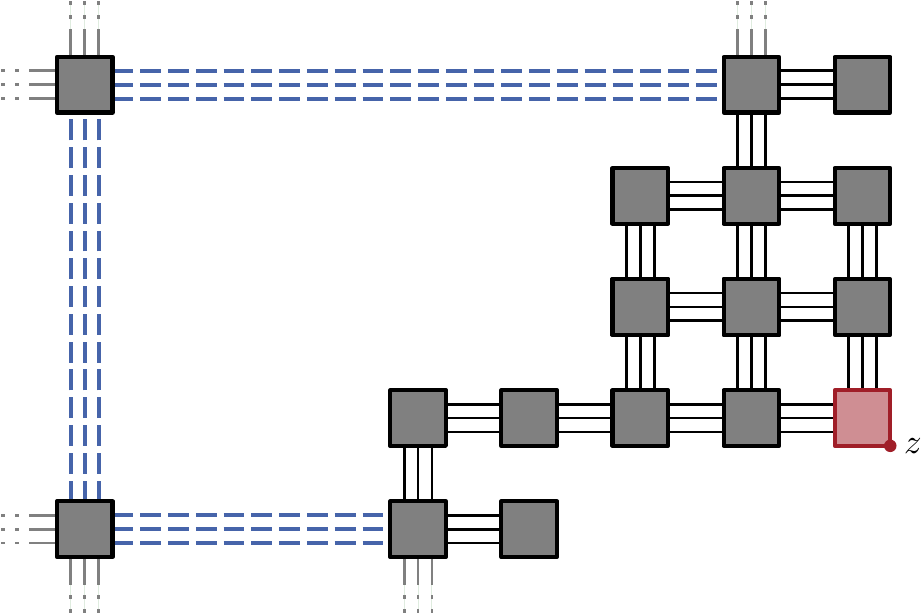} }
 \hfill
 \subfloat[]{ \label{fig:embed:gadgets-vertex}\includegraphics[scale=1.2]{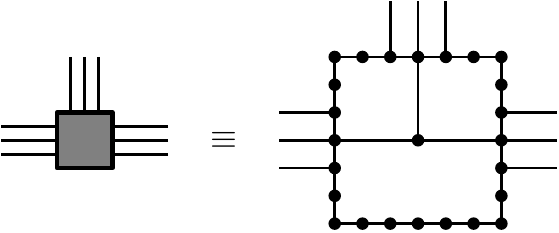} }
 \hspace{1cm}
 \subfloat[]{ \label{fig:embed:gadgets-edge}\includegraphics[scale=1.2]{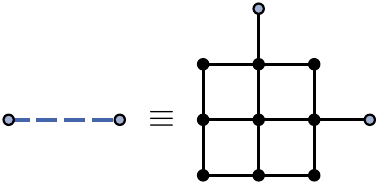} }
 \caption{(a) The clause gadget. The dashed blue edges must have exactly one bend. The light gray edges indicate where the clause gadget is connected to the remainder of $G$.
 (b) The subgraph each box stands for.
 (c) The gadget that represents edges that must have exactly one bend. By flipping the embedding one can achieve both bends to the left and to the right.}
 \label{fig:embed:gadgets}
\end{figure}
The idea behind the reduction is to rebuild the variable-clause graph of $\Phi$ and represent the truth values by edge bends.
Fig.~\ref{fig:embed:reduction} provides an example of this construction.
Bläsius et~al.\ introduce edges that must have exactly one bend (\emph{1-edges}) and represent true and false as bends in different directions.
In our figures, 1-edges are drawn as dashed blue lines (e.g., in Fig.~\ref{fig:embed:reduction-graph}).
However, instances of \textsc{Orthogonal Embeddability} cannot contain 1-edges and the gadget shown in Fig.~\ref{fig:embed:gadgets-edge} is used instead. Flipping the embedding of this gadget changes the direction of the bend.

We want to choose the port~$z$ on a part of $G$ that corresponds to one of the outermost clauses (e.g., on $c_1$ or $c_3$ in Fig.~\ref{fig:embed:reduction}).
The gadget they use to represent clauses, is shown in Fig.~\ref{fig:embed:gadgets-clause}.
Note that the boxes are not vertices but occurrences of the graph in Fig.~\ref{fig:embed:gadgets-vertex}.
Consider the gadget for one of the outermost clauses in the variable-clause graph.
We choose any degree-2 vertex of the box in the corner as out port~$z$ (cf.~Fig.~\ref{fig:embed:gadgets-clause}).
By the following lemma this vertex has the desired properties.

\begin{lemma}\label{lem:embed:existence_z}
Let $z$ be the port of $G$ as selected above. Then, exactly one of the following is true:
\begin{enumerate}
 \item There is an orthogonal embedding of $G$ in which $z$ lies on the outer face and the angle at $z$ in the outer face is at least $180\degree$.
 \item The graph $G$ does not admit any orthogonal embedding.
\end{enumerate}
\end{lemma}
\begin{proof}
If $G$ admits an orthogonal embedding, then $\Phi$ is satisfiable \cite{bbr14-kandinsky}. But then $G$ can be embedded like the variable-clause graph of $\Phi$.
In particular, $z$ lies on the outermost clause gadget. By the choice of $z$ this implies that $z$ also lies on the outer face of the whole graph $G$. Moreover, the angle at $z$ can be set to $270\degree$.
\end{proof}

Finding such a vertex $z$ was the missing piece in the reduction from \textsc{Planar Monotone 3-SAT} to \textsc{Orthogonal Embeddability}.

\begin{theorem}\label{thm:embed:NP_completeness}
\textsc{Ortho-Radial Embeddability} is $\mathcal{NP}$-complete.
\end{theorem}
\begin{proof}
Clearly, \textsc{Ortho-Radial Embeddability} lies in $\mathcal{NP}$.

The construction presented above first transforms an instance~$\Phi$ of \textsc{Planar Monotone 3-SAT} to an instance~$G$ of \textsc{Orthogonal Embeddability} and then to an instance $G'$ of \textsc{Ortho-Radial Embeddability}.
By Lemma~\ref{lem:embed:equivalence} the graph~$G'$ admits an ortho-radial embedding if and only if $G$ admits an orthogonal embedding with $z$ on the outer face.
According to Lemma~\ref{lem:embed:existence_z} this is in turn equivalent to the existence of an orthogonal embedding of $G$ without requiring $z$ to lie on the outer face.
Bläsius et~al.~\cite{bbr14-kandinsky} prove that $G$ admits an orthogonal embedding if and only if $\Phi$ is satisfiable.
In total, this implies that $\Phi$ and $G'$ are equivalent.
Furthermore, the reduction runs in polynomial time. As \textsc{Planar Monotone 3-SAT} is \NP-hard~\cite{bk12-optimal}, \textsc{Ortho-Radial Embeddability} is \NP-hard as well.
\end{proof}

One might wonder why we do not directly reduce \textsc{Orthogonal Embeddability} to \textsc{Ortho-Radial Embeddability} but instead start from \textsc{Planar Monotone 3-SAT}. This is due to the fact that we need to connect the triangles to a vertex of the instance~$G$ of \textsc{Orthogonal Embeddability}.
Therefore, we must find a vertex~$z$ on $G$ that can lie on the outer face. That is, if $G$ admits an orthogonal drawing, there is an orthogonal drawing of $G$ such that $z$ lies on the outer face.
For arbitrary instances, we cannot identify such a vertex easily. For instances created by the reduction from \textsc{Planar Monotone 3-SAT} however, we can exploit the structure to find a suitable vertex.

\section{Conclusion}
In this paper we considered orthogonal drawings of graphs on
cylinders. Our main result is a characterization of the 4-plane graphs that can be drawn bend-free on a cylinder in terms of a combinatorial description of such drawings. These ortho-radial representations determine all angles in the drawing without fixing any lengths, and thus are a natural extension of Tamassia's orthogonal representations. However, compared to those, the proof that every valid ortho-radial representation has a corresponding drawing is significantly more involved. The reason for this is the more global nature of the additional property required to deal with the cyclic dimension of the cylinder.

Our ortho-radial representations establish the existence of an ortho-radial TSM framework in the sense that they are a combinatiorial description of the graph serving as interface between the ``Shape'' and ``Metrics'' step.

For rectangular 4-plane graphs, we gave an algorithm producing a drawing from a valid ortho-radial representation. Our proof reducing the drawing of general 4-plane graphs with a valid ortho-radial representation to the case of rectangular 4-plane graphs is constructive; however, it requires checking for the violation of our additional consistency criterion. It is an open question whether this condition can be checked in polynomial time. These algorithms correspond to the ``Metrics'' step in a TSM framework for ortho-radial drawings.

Since the additional property of the characterization is non-local (it is
based on paths through the whole graph), the original flow network by Tamassia cannot be easily adapted to compute the ``Shape'' step of an ortho-radial TSM framework. We leave this as an open question.

 \newcommand{\bibsoda}[2]{Proceedings of the #1 Annual ACM-SIAM Symposium on
  Discrete Algorithms (SODA'#2)} \newcommand{\bibgd}[2]{Proceedings of the #1
  International Symposium on Graph Drawing (GD'#2)}
  \newcommand{\bibfocs}[2]{Proceedings of the IEEE #1 Annual Symposium on
  Foundations of Computer Science (FOCS'#2)}
  \newcommand{\bibinfovis}[1]{Proceedings of the IEEE Symposium on Information
  Visualization (InfoVis'#1)} \newcommand{\bibvis}[1]{Proceedings of the IEEE
  Conference on Visualization (Vis'#1)} \newcommand{\bibpvis}[1]{Proceedings of
  the IEEE Pacific Visualisation Symposium (PacificVis'#1)}
  \newcommand{\bibsoftvis}[2]{Proceedings of the #1 ACM Symposium on Software
  Visualization (SoftVis'#2)} \newcommand{\bibeurocg}[2]{Proceedings of the #1
  European Workshop on Computational Geometry (EuroCG'#2)}
  \newcommand{\bibsocg}[2]{Proceedings of the #1 Annual Symposium on
  Computational Geometry (SoCG'#2)} \newcommand{\bibwads}[2]{Proceedings of the
  #1 International Symposium on Algorithms and Data Structures (WADS'#2)}
  \newcommand{\bibwg}[2]{Proceedings of the #1 Workshop on Graph-Theoretic
  Concepts in Computer Science (WG'#2)} \newcommand{\bibgta}{Proceedings of the
  Conference at Graph Theory and Applications}
  \newcommand{\bibisaac}[2]{Proceedings of the #1 International Symposium on
  Algorithms and Computation (ISAAC'#2)} \newcommand{\bibcocoon}[2]{Proceedings
  of the #1 Annual International Conference on Computing and Combinatorics
  (COCOON'#2)} \newcommand{\bibtamc}[2]{Proceedings of the #1 Annual Conference
  on Theory and Applications of Models of Computation (TAMC'#2)}
  \newcommand{\bibicalp}[2]{Proceedings of the #1 Int. Colloquium on Automata,
  Languages and Programming (ICALP'#2)}

\end{document}